%% file: main.tex
\documentclass[aps,twocolumn,notitlepage,amsmath,amssymb,superscriptaddress]{revtex4-2}

\usepackage{dsfont}
\usepackage{makecell}
\usepackage{tabularx}
\usepackage{amsmath}
\usepackage{amssymb}
\usepackage{amsfonts}
\usepackage{amsthm}
\usepackage{mathtools}
\usepackage{mathrsfs}

\usepackage{MnSymbol}
\usepackage{braket}  

\usepackage{graphicx}
\usepackage{graphics}
\usepackage{booktabs}
\usepackage{dcolumn}
\usepackage{multirow}

\usepackage{times}
\usepackage{gensymb}
\usepackage{enumerate}
\usepackage{changes}

\usepackage{color}
\usepackage{xcolor}

\definecolor{myurlcolor}{rgb}{0,0,0.7}
\definecolor{myrefcolor}{rgb}{0.1,0,0.9}

\definecolor{THc}{rgb}{0.9,0.3,0.2}

\usepackage[
	breaklinks,
	pdftex,
	colorlinks=true, 
	linkcolor=myrefcolor,
	citecolor=myurlcolor,
	urlcolor=myurlcolor
]{hyperref}

\newcommand{\ignore}[1]{}

\newtheorem{lemma}{Lemma}
\newtheorem{theorem}{Theorem}

\newtheorem{proposition}{Proposition}
\newtheorem*{theorem*}{Theorem}


\begin{document}

\title{Quantum Cubature Codes}

\author{Yaoling Yang}
\email{yaoling.space@gmail.com}
\affiliation{Quantum Innovation Centre (Q.InC), Agency for Science, Technology and Research (A*STAR), 2 Fusionopolis Way, Innovis
\#08-03, Singapore 138634, Republic of Singapore}

\author{Andrew Tanggara}
\email{andrew.tanggara@gmail.com}
\affiliation{Centre for Quantum Technologies, National University of Singapore, 3 Science Drive 2, Singapore 117543}
\affiliation{Nanyang Quantum Hub, School of Physical and Mathematical Sciences, Nanyang Technological University, Singapore 639673}

\author{Tobias Haug}
\email{tobias.haug@u.nus.edu}
\affiliation{Quantum Research Center, Technology Innovation Institute, Abu Dhabi, UAE}

\author{Kishor Bharti}
\email{kishor.bharti1@gmail.com}
\affiliation{Institute of High Performance Computing (IHPC), Agency for Science, Technology and Research (A*STAR), 1 Fusionopolis
Way, \#16-16 Connexis, Singapore 138632, Republic of Singapore}
\affiliation{Quantum Innovation Centre (Q.InC), Agency for Science, Technology and Research (A*STAR), 2 Fusionopolis Way, Innovis \#08-03, Singapore 138634, Republic of Singapore}

\date{\today}

\begin{abstract}
Bosonic codes utilize the infinite-dimensional Hilbert space of harmonic oscillators to encode quantum information, offering a hardware-efficient approach to quantum error correction.  Designing these codes requires precise geometric arrangements of quantum states in the phase space. Here, we introduce Quantum Cubature Codes (QCCs), a powerful and generalized framework for constructing bosonic codes based on superpositions of coherent states. This formalism utilizes cubature formulas from multivariate approximation theory, which connect the continuous geometry of the phase space to discrete, weighted point sets, ensuring the conditions for error correction are met. We demonstrate that this framework provides a unifying perspective, revealing that well-established codes, such as cat codes and the recently proposed quantum spherical codes (QSCs), are specific instances of QCCs corresponding to uniform weights on a single energy shell. The QCC formalism unlocks a vast new design space, encompassing non-uniform superpositions and multi-shell configurations. We leverage this framework to discover several new families of codes derived from Euclidean designs, allowing for greater geometric separation between logical states, which correlates with improved performance under photon loss. Numerical simulations under a pure-loss channel show that these novel multi-shell QCCs can outperform their single-shell counterparts by maximizing geometric separation with optimal energy at fixed pure-loss rate.
\end{abstract}

\maketitle

\section{Introduction}
Quantum error correction (QEC) is indispensable for realizing fault-tolerant quantum computation, protecting fragile quantum information from noise and decoherence~\cite{PhysRevA.52.R2493, PhysRevLett.77.793, aharonov1997fault, PhysRevA.86.032324, RevModPhys.87.307}. While traditional QEC focuses on discrete qubit systems, bosonic codes offer a compelling alternative~\cite{albert2022bosonic,Terhal_2020,noh2021quantumcomputationcommunicationbosonic,Cai_2021,joshi2021quantum, cochrane1999macroscopically, Mirrahimi_2014, leghtas2015confining}. By encoding information in the infinite-dimensional Hilbert space of bosonic modes, such as electromagnetic modes or mechanical resonators, bosonic codes can achieve high hardware efficiency~\cite{Niu_2018,lemonde2024hardwareefficientfaulttolerantquantum,mori2024hardwareefficientbosonicquantumcomputing,Putterman_2025}. A single physical oscillator can host a logical qubit and provide mechanisms to correct dominant errors like photon loss. 

The landscape of bosonic codes is rich, featuring diverse strategies like the famous Gottesman-Kitaev-Preskill (GKP) codes~\cite{Gottesman_2001,Conrad_GKP,Brady_2024,fluhmann_encoding_2019,campagne2020quantum} and binomial codes~\cite{michael2016new,hu2019quantum}. Among the most promising are codes based on superpositions of coherent states~\cite{cochrane1999macroscopically, Mirrahimi_2014, leghtas2015confining, leghtas2013hardware, mirrahimi2014dynamically,guillaud2023quantum, PhysRevX.9.041009}. These codes leverage the geometry of phase space, arranging coherent states in specific constellations to ensure robustness against errors. However, the design of these constellations has often relied on geometric intuition and highly symmetric arrangements, rather than a systematic mathematical framework.

In this work, we introduce Quantum Cubature Codes (QCCs), a comprehensive and mathematically powerful formalism for the systematic construction of coherent-state bosonic codes. QCCs are grounded in the classical theory of cubature formulas~\cite{cools1992survey,cools2003encyclopaedia,stroud1971approximate,sawa2019euclidean,bannai2010cubature}. A cubature formula is a method for numerical integration that achieves exact results for polynomials up to a certain degree by using a weighted summation over a discrete set of points. By linking the requirements of QEC to the mathematics of cubature formulas, the QCC formalism provides a rigorous method for designing bosonic codes. Our approach offers a unifying perspective on coherent-state codes. We show that prominent existing codes are, in fact, special cases of the QCC framework. Cat codes~\cite{cochrane1999macroscopically, leghtas2013hardware, mirrahimi2014dynamically,guillaud2023quantum, PhysRevX.9.041009} and the recently introduced Quantum Spherical Codes (QSCs)~\cite{jain2024quantum} emerge as QCCs characterized by uniform weights and constellations confined to a single energy shell (a sphere in the phase space). Crucially, the QCC formalism generalizes these concepts, unlocking a much broader design space. It naturally incorporates codes with non-uniform superposition weights and, significantly, multi-shell configurations where coherent states are superposed across concentric spheres with different energies.

We leverage this new framework to discover several novel families of bosonic codes based on Euclidean designs, which can be viewed as a special class of cubature formulas~\cite{bannai2010cubature, sawa2019euclidean}. We demonstrate that the additional flexibility provided by the multi-shell structure allows for greater geometric separation between logical states at a fixed mean photon number compared to existing QSCs. Through numerical benchmarking under a bosonic pure-loss channel, we show that these newly discovered QCCs can outperform existing single-shell QSCs, in terms of entanglement fidelity. The QCC formalism thus provides a systematic pathway to optimized bosonic error correction.

\section{Coherent-state constellation codes}

For an $n$-mode bosonic system with annihilation operators $\{a_j\}$ and creation operators $\{a_j^\dagger\}$ for $j=1,\ldots,n$, a multimode coherent state is
\begin{equation}
\ket{\boldsymbol{\alpha}}=\bigotimes_{j=1}^n\ket{\alpha_j},
\end{equation}
where $\boldsymbol{\alpha}=(\alpha_1,\ldots,\alpha_n)\in\mathbb{C}^n$. These states are eigenstates of the annihilation operators, satisfying $a_j\ket{\boldsymbol{\alpha}}=\alpha_j\ket{\boldsymbol{\alpha}}$. Typically, coherent states are not orthogonal to each other, with their overlap given by
\begin{equation}\label{eq:coherent-overlap}
|\langle\boldsymbol{\alpha}\mid \boldsymbol{\beta}\rangle|
=\exp\!\Big(-\tfrac12\|\boldsymbol{\alpha}-\boldsymbol{\beta}\|^2\Big)\,,
\end{equation}
which is exponentially suppressed by the squared Euclidean distance between the amplitudes. With sufficiently large separation, distinct coherent states become nearly orthogonal, namely $\langle\boldsymbol{\alpha}| \boldsymbol{\beta}\rangle \approx 0$.

Consider a general coherent-state based bosonic code. The $k$-th logical codeword (constellation) $\ket{\mathcal{C}_k}$ can be defined as
\begin{equation}
\label{eq:general_const}
    \ket{\mathcal{C}_k}
    = \sum_{\boldsymbol{\alpha}\in\mathcal{C}_k} c_{\boldsymbol{\alpha}}\,\ket{\boldsymbol{\alpha}},
\end{equation}
where $\mathcal{C}_k{\subset} \mathbb{C}^n$ denotes a finite set of complex amplitude vectors, and $\{c_{\boldsymbol{\alpha}}\}$ are normalized coefficients. For such a construction to constitute a valid quantum code, it must satisfy the Knill-Laflamme~(KL) condition~\cite{knill1997theory}. Specifically, given a code subspace spanned by an orthonormal basis of logical states $\{\ket{\mathcal{C}_k}\}$, a set of error operators $\{E_\mu\}$ is correctable if and only if
\begin{equation}\label{eq:kl-correction}
\langle \mathcal{C}_k \mid E_{\mu}^\dagger E_{\nu} \mid \mathcal{C}_\ell \rangle = \delta_{k\ell} \lambda_{\mu\nu},
\end{equation}
where $\lambda_{\mu\nu}$ are elements of a Hermitian matrix independent of the logical state indices $k, \ell$, and $\delta_{k\ell}$ is the Kronecker delta. On the other hand, these error operators are detectable if and only if
\begin{equation}
\label{eq:KL-detection}
\langle \mathcal{C}_k \mid E_\mu \mid \mathcal{C}_\ell \rangle = \delta_{k\ell} \lambda_\mu,
\end{equation}
where $\lambda_\mu$ is a constant. The error set $\{E_\mu\}$ for coherent-state constellation codes mainly comprises two significant classes~\cite{turchette2000decoherence,jain2024quantum}. The first is angular dephasing, which arises from passive linear-optical transformations represented by unitary rotations $R\in\mathrm{U}(n)$ that maps $\ket{\boldsymbol{\alpha}}\mapsto\ket{R\boldsymbol{\alpha}}$. These transformations rotate $\boldsymbol{\alpha}$ while preserving its norm $\|\boldsymbol{\alpha}\|$. For small angular perturbations, such errors are effectively suppressed by large geometric separation in the large-energy limit, where distinct coherent states become nearly orthogonal. The second class is excitation-changing errors, which alter the photon number in each mode and are captured by the \emph{ladder operators}:
\begin{equation}L_{\mathbf{p},\mathbf{q}}(\boldsymbol{a}^\dagger,\boldsymbol{a})= \prod_{j=1}^n (a_j^\dagger)^{p_j} a_j^{q_j},
\end{equation}
where $\mathbf{p},\mathbf{q}\in\mathbb{N}^{n}$ are multi-indices specifying the number of photon creation and annihilation events on the $n$ modes, and $\boldsymbol{a}=(a_1,a_2,\ldots,a_n)$ (respectively, $\boldsymbol{a}^\dagger$) denote the annihilation (creation) operators on the $n$ modes. Products of passive linear-optical transformations and ladder operators form a complete operator basis. Consequently, the Kraus operators of any physical noise channel can be expanded as linear combinations of such products, often as an infinite series~\cite{grimsmo2020quantum}. For bosonic codes to protect against such excitation-changing errors, the logical states must satisfy the KL conditions. For example, for pure-loss errors, the error operators $E_\nu$ can be expressed as polynomials in the annihilation operators, so every product $E_\mu^\dagger E_\nu$ appearing in Eq.~\eqref{eq:kl-correction} can be expanded as a linear combination of monomials $L_{\mathbf{p},\mathbf{q}}(\boldsymbol{a}^\dagger, \boldsymbol{a})$, the KL condition reduces to
\begin{equation}
\bra{\mathcal{C}_k}\,  E_{\mu}^\dagger E_{\nu}\,\ket{\mathcal{C}_\ell} \propto  \sum_{\substack{\boldsymbol{\alpha} \in \mathcal{C}_k \\ \boldsymbol{\beta} \in \mathcal{C}_\ell}}c_{\boldsymbol\alpha}^*c_{\boldsymbol\beta}f_{(\mathbf{p},\mathbf{q})}(\boldsymbol{\alpha}^*,\boldsymbol{\beta})\langle{\boldsymbol{\alpha}}|\boldsymbol{\beta}\rangle,
\end{equation}
where we denote $ f_{(\mathbf{p},\mathbf{q})}(\boldsymbol{\alpha}^*,\boldsymbol{\beta}) = \prod_{j=1}^n (\alpha_j^*)^{p_j} \beta_j^{q_j}$, and $c_{\boldsymbol\alpha}^*$, $c_{\boldsymbol\beta}$ are corresponding coefficients of coherent states. In the large-energy limit, the approximate orthogonality of distinct coherent states ensures that off-diagonal KL condition terms ($k\neq\ell$) vanish, while the diagonal terms simplify to
\begin{equation}
\label{eq:moment-matching}
\langle \mathcal{C}_k \mid E_{\mu}^\dagger E_{\nu}  \mid \mathcal{C}_k \rangle 
\propto \sum_{\boldsymbol{\alpha} \in \mathcal{C}_k} |c_{\boldsymbol\alpha}|^2f_{(\mathbf{p},\mathbf{q})}(\boldsymbol\alpha^*,\boldsymbol\alpha).
\end{equation}
Thus, the physical requirement for error correction translates to a mathematical constraint: for every pair of multi-indices $(\mathbf{p},\mathbf{q})$ in the chosen error set, all constellations $\{\mathcal{C}_k\}$ must yield the same weighted average of $f_{(\mathbf{p},\mathbf{q})}(\boldsymbol\alpha^*,\boldsymbol\alpha)$, independent of $k$.

\section{Cubature code formalism}

The problem of finding coherent state constellations satisfying the KL condition has a direct connection to a mathematical framework in multivariate approximation theory: \emph{cubature formulas}. A cubature formula of degree $t$ for a measure $\sigma$ over a domain $\Omega \subset \mathbb{R}^D$ consists of a finite set of $N$ points $\mathcal{V} = \{\boldsymbol{\alpha}_i\}_{i=1}^N \subset \Omega$ and corresponding weights $\mathcal{W}=\{w_i\}_{i=1}^N$ such that
\begin{equation}\label{eq:cubature}
\sum_{i=1}^{N} w_i f(\boldsymbol{\alpha}_i) = \int_{\Omega} f(\boldsymbol{\alpha})\,\mathrm{d}\sigma(\boldsymbol{\alpha})
\end{equation}
holds for all polynomials $f$ of degree at most $t$. In essence, cubature formulas provide exact polynomial integration via weighted discrete summation.
By selecting constellation $\mathcal{C}_k$ from cubature formulas that compute an identical integral, the discrete summation in Eq.~\eqref{eq:moment-matching} equals the domain integration, namely
\begin{align}
    \sum_{\boldsymbol{\alpha} \in \mathcal{C}_k} |c_{\boldsymbol\alpha}|^2f_{(\mathbf{p},\mathbf{q})}(\boldsymbol\alpha^*,\boldsymbol\alpha) &=  \int_{\Omega} f_{(\mathbf{p},\mathbf{q})}(\boldsymbol\alpha^*,\boldsymbol\alpha) \,\mathrm{d}\sigma(\boldsymbol{\alpha}) \nonumber \\ &=   \lambda_{(\mathbf{p},\mathbf{q})},
    \end{align}
where $\lambda_{(\mathbf{p},\mathbf{q})}$ is a constant independent of the logical state $\ket{\mathcal{C}_k}$, thereby providing a sufficient condition for satisfying the corresponding KL conditions.

We define a QCC as a bosonic code whose logical basis states are built from \emph{weighted constellations} $(\mathcal{V}_k,\mathcal{W}_k)$ that realize degree-$t$ cubature formulas on a common domain $\Omega$. Concretely, each logical state index $k$ is associated with a finite point set (cubature points) $\mathcal{V}_k{=}\{\boldsymbol{\alpha}_i\}_{i=1}^{N_k} \subset \Omega$ together with normalized weights $\mathcal{W}_k{=}\{w_{\boldsymbol{\alpha}}^{(k)}\}_{\boldsymbol{\alpha}\in\mathcal{V}_k}$ satisfying $\sum_{\boldsymbol{\alpha}\in\mathcal{V}_k} w_{\boldsymbol{\alpha}}^{(k)}{=}1$, where all weights are taken to be positive. For each logical index $k$, the pair $(\mathcal{V}_k,\mathcal{W}_k)$ is assumed to realize an equivalent degree-$t$ cubature formula, meaning that all such pairs yield the same integral values for every polynomial of total degree at most $t$. In the case of a spherically symmetric integral, different pairs $(\mathcal{V}_k,\mathcal{W}_k)$ can be generated by applying orthogonal rotations to a reference point set $\mathcal{V}_0$. Such rotations may be chosen either to maximize the minimum distances between constellation points belonging to distinct $\mathcal{V}_k$ or to preserve a desired symmetry. More generally, distinct logical constellations can be constructed from different cubature formulas that compute the same integral over a common domain, thereby ensuring that the KL conditions to hold. The corresponding logical codeword is encoded as the coherent-state superposition
\begin{equation}
\ket{\mathcal{C}_k} \propto \sum_{\boldsymbol{\alpha} \in \mathcal{V}_k} \sqrt{w_{\boldsymbol{\alpha}}^{(k)}} |\boldsymbol{\alpha}\rangle,
\end{equation}
where $|\boldsymbol{\alpha}\rangle$ denotes an $n$-mode coherent state with amplitude vector $\boldsymbol{\alpha}$. Notably, although cubature formulas are typically defined on $\mathbb{R}^D$, one can introduce a complex embedding that maps real nodes to complex amplitudes $\boldsymbol{\alpha}$ while preserving the cubature identity: the weighted sum over the embedded points still reproduces the same integral as in the original real-valued cubature formula; see the Supplemental Material (SM) for details~\cite{SM_QCC}. In this work we mainly consider the case where $D$ is even, and we use the standard real-to-complex embedding $\Phi:\mathbb{R}^D \to \mathbb{C}^{D/2}$ defined by
\begin{equation}\label{eq:complex_embedding}
\Phi(x_1,\ldots,x_D) = (x_1 + i x_2, x_3 + i x_4, \ldots, x_{D-1} + i x_D),
\end{equation}
which identifies real coordinates with complex coherent-state amplitudes.

To see that QCCs form a valid bosonic code, one needs to examine two noise models. For angular dephasing noise, which shifts a coherent state $\ket{\boldsymbol{\alpha}}$ to $\ket{R\boldsymbol{\alpha}}$, this effect becomes (approximately) detectable in the large-energy limit as long as $R\boldsymbol{\alpha}$ does not lie in the code constellation $\{\mathcal{V}_k\}_k$. Consequently, the minimum angular separation between amplitude vectors $\boldsymbol{\alpha}$ that belong to the same shell characterizes the code’s protection against angular dephasing noise.  

The second type is excitation-changing errors; the cubature-based structure implies the following error correction capability:
\begin{theorem} \label{QCC_thm} 
In the large-energy limit, where distinct coherent states become approximately orthogonal, any QCC constructed from degree-$t$ cubature formulas can correct all photon-loss errors $L_{\mathbf{0},\mathbf{q}}=\prod_{j=1}^n a_j^{q_j}$
with $|\mathbf{q}| \le \lfloor t/2\rfloor$, and can detect all photon loss-gain errors 
$L_{\mathbf{p},\mathbf{q}}=\prod_{j=1}^n (a_j^\dagger)^{p_j} a_j^{q_j}$ with $|\mathbf{p}+\mathbf{q}| \le t$.
\end{theorem}
The proof can be found in the Supplemental Material (SM)~\cite{SM_QCC}. This theorem shows that high-degree cubature formulas translate directly into protection against higher-order photon-loss errors, while achieving a larger $t$ typically requires more cubature points in the underlying cubature formula and hence more coherent states in each logical superposition. Conversely, any desired level of loss protection imposes a minimum ``constellation size'', namely one cannot realize a degree-$t$ QCC with fewer coherent-state components than allowed by well-known fundamental bounds on cubature formulas. These trade-offs can be summarized as follows:
\begin{proposition}[Bounds on Constellation Size]
\label{prop:const_size} 
For any QCC constructed from degree-$t$ cubature formulas, there exist bounds $N_{\text{min}}$ and $N_{\text{exist}}$ such that:
\begin{itemize}
\item[(i)] 
Each logical state requires at least $N_{\text{min}}$ superposed coherent states.
\item[(ii)] 
There exists a code construction with at most $N_{\text{exist}}$ superposed coherent states per logical state.
\end{itemize}
\end{proposition}
In general, the exact values of $N_{\text{min}}$ and $N_{\text{exist}}$ are not known in closed form. 
Analytic lower bounds from cubature theory, such as Fisher-type bound~\cite{delsarte_spherical_1977} and the M\"oller bound~\cite{moller1979lower} (both are rarely saturable), provide computable estimates for $N_{\text{min}}$, while Tchakaloff's theorem~\cite{tchakaloff1957formules,sawa2019euclidean} yields an analytic existence upper bound on $N_{\text{exist}}$. 
These bounds are determined by the dimensions of the relevant polynomial spaces on $\Omega$ (see SM for details), and they provide a natural benchmark for assessing cubature formulas in terms of their number of cubature points. Cubature formulas that saturate any of these bounds are called \textit{minimal} cubature formulas, and the associated point sets can be regarded as \emph{tight} designs~\cite{bannai2015survey}. We thus refer to QCCs whose logical constellation sizes saturate the analytic lower bound (Fisher-type or M\"oller bound) as \emph{tight} QCCs. Typically, when the underlying constellations are built from tight designs, saturation of those lower bounds imposes strong algebraic and geometric constraints, including high symmetry, few distinct inner products, and in many cases complete classification~\cite{delsarte_spherical_1977,bannai2009survey}. This makes constellations based on tight designs the natural building blocks for high-quality QCCs. Later, we will make some remarks on how the bounds on the constellation size imposed by Proposition~\ref{prop:const_size} may give rise to trade-offs between the dimension of a QCC and its error-correcting capabilities, as quantified by its parameters.

\textbf{Code parameters.} 
As coherent-state constellation codes, QCCs inherit a similar distance
characterization to QSCs. To quantify protection against excitation-changing errors, we adopt the distance parameters $\langle t_\downarrow, d_\updownarrow, d_\downarrow \rangle$ from \cite{jain2024quantum}: $(t_\downarrow-1)$ denotes the number of correctable photon-loss errors, $(d_\downarrow-1)$ denotes the number of detectable pure-loss errors ($L_{\mathbf{0},\mathbf{q}}$), and $(d_\updownarrow-1)$ denotes the number of a detectable photon-loss-gain errors ($L_{\mathbf{p},\mathbf{q}}$). It follows from Theorem~\ref{QCC_thm} that we have $t_\downarrow \ge \lfloor t/2 \rfloor$ and $d_\updownarrow \ge t+1$.

Beyond protection against excitation-changing errors, the geometric separation between constellation points characterizes the code's robustness against additional error channels. This separation decomposes into two components: angular separation (along a fixed shell), which governs resilience to dephasing errors, and radial separation (between different shells), which provides additional protection against excitation-changing errors. We quantify the overall geometric separation using the global nearest-neighbor squared Euclidean distance (also known as \emph{resolution}):
\begin{equation}
    d_E = \min_{\boldsymbol{\alpha},\boldsymbol{\beta}} \|\boldsymbol{\alpha} - \boldsymbol{\beta}\|^2.
\end{equation}
where the minimization is over all pair of distinct points $\boldsymbol{\alpha},\boldsymbol{\beta}$ in the union of the constellations $S = \bigcup_k \mathcal{V}_k$.
To enable a fair comparison for $d_E$ between single-shell QSCs
and the multi-shell (QCCs) generalization, we fix the mean photon number
of each logical code state. Given a logical constellation 
$\mathcal{C}_k$, 
we define the effective energy as
\begin{equation}
\bar n := \frac{\sum_{\boldsymbol{\alpha} \in \mathcal{C}_k} w^{(k)}_{\boldsymbol{\alpha}} \|\boldsymbol{\alpha}\|^2}{\sum_{\boldsymbol{\alpha} \in \mathcal{C}_k} w_{\boldsymbol{\alpha}}^{(k)}}.
\end{equation}
For single-shell QSCs we choose $\|\boldsymbol{\alpha}\|^2 = 1$ for all
$\boldsymbol{\alpha} \in \mathcal{C}_k$ so that $\bar n_{\mathrm{QSC}} = 1$. For multi-shell constellations, coherent states on shell $s$ possess energy $\mathcal{R}_s^2$. To enable fair comparison across different code structures, we introduce a global scaling factor $\lambda$ such that $r_s = \lambda \mathcal{R}_s$. The effective mean photon number of the resulting QCC becomes \begin{equation} \label{eq:energy_QCC}
\bar n_{\mathrm{QCC}}(\lambda) =  \sum_s M_s r_s^2= \lambda^2 \sum_s M_s \mathcal{R}_s^2, \end{equation} where $M_s = \sum_{\boldsymbol{\alpha} \in \mathcal{C}_k^{(r_s)}} w_{\boldsymbol{\alpha}}$ denotes the total weight assigned to shell $s$, accounting for potentially non-uniform intra-shell weights. We fix $\lambda$ by requiring $\bar n_{\mathrm{QCC}}(\lambda) = 1$, thereby ensuring that all codes are evaluated under the same energy constraint.

We adopt the notation of~\cite{jain2024quantum} and denote a QCC as $(\!(n, K, d_{E} ,\langle t_\downarrow, d_\updownarrow, d_\downarrow \rangle)\!)$, representing an $n$-mode code encoding $K$ logical states with resolution $d_E$ and $\langle t_\downarrow, d_\updownarrow, d_\downarrow\rangle$ against excitation-changing errors. For QCCs, the resolution $d_E$ measures the distance across multiple shells, extending beyond the single-shell structure of QSCs.

Finally, we remark how the bounds on constellation size given in Proposition~\ref{prop:const_size} may impose some restrictions on the parameters of QCCs constructed from degree-$t$ cubature formulas.
While the degree $t$ imposes lower-bounds $t_\downarrow$ and $d_\updownarrow$ (due to Theorem~\ref{QCC_thm}), how the number of modes $n$, codespace dimension $K$, and resolution $d_E$ are related to $t$ are more indirect.
For a degree-$t$ QCC with logical subspace of dimension $K$ over $n$ modes, Proposition~\ref{prop:const_size} forces the number of points in each of the $K$ constellations defining the codewords to be at least $N_{\min}^{(t)}$ (here we write $N_{\min}^{(t)}$ instead of $N_{\min}$ to emphasize its dependence on $t$ explicitly).
Thus the total number of points for this QCC is $KN_{\min}^{(t)}$.
Imposing an energy constraint on this QCC by limiting the amplitude as $|\alpha|\leq \xi$ for each point~\footnote{Otherwise, we can make $K$ and $d_E$ arbitrarily large if we can use the entire phase-space.}, the lower bound on the constellation size $N_{\min}^{(t)}$ now puts some restriction on $n$, $K$, and $d_E$.
Geometrically, this can be viewed as a packing problem of ``fitting'' $KN_{\min}^{(t)}$ many $n$-dimensional balls inside an $(n-1)$-dimensional hypersphere with radius $\xi$.
For example, if $K=1$, $n=2$, $N_{\min}^{(t)}=2$, and $\xi=1$, we can obtain an upper bound on $d_E\leq4$ by placing $KN_{\min}^{(t)}=2$ points on two extreme ends of the unit circle on the $x$-axis.
However, once we increase $K=2$, the resolution upper bound becomes $d_E\le2$ by placing $KN_{\min}^{(t)}=4$ points at the four extreme points of the unit circle on the x and y axes.
Note that this intuition resembles the quantum Hamming bound on discrete-variable QECs~\cite{Gottesman_1996,aly2007notequantumhammingbound} that imposes a tradeoff between the number of physical qubits $n$, logical qubits $k$, and distance $d$.

\section{Examples}\label{sec:examples}

QCCs arise from cubature formulas on various manifolds and encompass cat codes and QSCs as special cases. A family of $2m$-component cat qubits can be obtained based on a class of cubature formulas on $\mathbb{S}^1$~\cite[p.~296]{stroud1971approximate}, where the logical constellations are given by
\begin{equation}\label{eq:QSC8}
    \mathcal{C}_{\text{cat}} = \bigcup_{j=1}^{m} \left\{\left( \cos \frac{2(j-1)\pi}{m}, \sin \frac{2(j-1)\pi}{m} \right) \right\},
\end{equation}
with uniform superposition. For $m=2$, the constellation $\mathcal{C}_{\text{cat}}$ consists of two antipodal points. Two logical codewords are obtained from two such constellations related by a $\pi/2$ rotation, yielding the standard four-component cat code. This code corrects single-photon loss with $d_E{=}2$ between constellation points.

Cubature formulas in higher-dimensional spaces $\mathbb{R}^D$ naturally yield multi-mode bosonic codes. As a special case, cubature formulas with uniformly weighted points distributed on a single sphere correspond to spherical designs~\cite{roy2014complex,bannai2021algebraic}, and the resulting codes are QSCs. For example, consider the vertices of a $D$-dimensional hypercube ($D$-cube) on unit sphere:
\begin{equation}
\mathcal{C}_{D\text{-}\mathrm{cube}}=\tfrac{1}{\sqrt{D}}\Big\{(\epsilon_1,\ldots,\epsilon_D):\ \epsilon_j\in\{\pm1\}\Big\},
\end{equation}
and a $D$-dimensional cross-polytope ($D$-orthoplex):
\begin{equation}
\mathcal{C}_{D\text{-}\mathrm{orth}}=\big\{(\pm1,0,\ldots,0)\cup \text{perms.}\big\}.
\end{equation}
These two configurations define classes of uniform-weight cubature formulas, with points distributed on a single sphere~\cite[Sec.~8.6 Ex.3-1, Ex.3-2]{stroud1971approximate}. Each set of points can serve as a logical constellation and, through appropriate rotation operations, yield QSCs. As a specific example in $\mathbb{R}^4$, consider logical constellations given by the $16$-cell (4-orthoplex) and the $8$-cell (4-cube). Three copies of $\mathcal{C}_{16\text{-}\mathrm{cell}}$ yield a 2-mode qutrit QSC $(\!(2,3,1,\langle 2,4,4\rangle)\!)$ that detects 2 photon losses with $d_E=1$. Alternatively, combining $\mathcal{C}_{16\text{-}\mathrm{cell}}$ with $\mathcal{C}_{8\text{-}\mathrm{cell}}$ gives a 2-mode qubit QSC $(\!(2,2,1,\langle 2,4,4\rangle)\!)$ with distinct logical constellations yet identical error correction performance.

Furthermore, a family of (weighted) QSCs can be constructed from cubature formula~\cite[Sec.~8.6, Ex.~5-2]{stroud1971approximate} by taking the union of two constellations
\begin{equation}\label{eq:weighted_QSC}
\mathcal{C}_{\text{logical}} = \mathcal{C}_{D\text{-}\mathrm{cube}}\cup\mathcal{C}_{D\text{-}\mathrm{orth}},
\end{equation}
and assigning weights to coherent states according to
\begin{equation}
|w_{\boldsymbol{\alpha}}|^2=
\begin{cases}
\dfrac{D}{2^{D}(D+2)}, & \boldsymbol{\alpha}\in\mathcal{C}_{D\text{-}\mathrm{cube}},\\[8pt]
\dfrac{1}{D(D+2)}, & \boldsymbol{\alpha}\in\mathcal{C}_{D\text{-}\mathrm{orth}}.
\end{cases}
\end{equation}
These cubature formulas yield QSCs whose superposition weights are uniform or non-uniform depending on $D$. For even dimensions, the standard real-to-complex embedding $\Phi: \mathbb{R}^{D}{\to}\mathbb{C}^{D/2}$ defined in Eq.~\eqref{eq:complex_embedding} maps cubature points to bosonic constellation codes. Examples include: a $(\!(1,2,4\sin^2(\pi/16),\langle 8,8,8\rangle)\!)$ 16-leg cat code ($D{=}2$) with octagonal logical constellations, a $(\!(2,2,2{-}\sqrt{2},\langle 5,6,12\rangle)\!)$ uniformly weighted QSC ($D{=}4$) with logical constellations forming the 24-cell polytope (union of the 16-cell and 8-cell vertices), and weighted QSCs for $D{\ge}6$ such as $(\!(3,2,2{-}\sqrt{2},\langle 5,6,12\rangle)\!)$ ($D{=}6$).

Beyond single-sphere configurations, cubature formulas naturally accommodate multi-shell QCCs, where coherent states are distributed over concentric spheres corresponding to different energy levels. This multi-shell setting is closely connected to \emph{Euclidean designs}~\cite{bannai2015survey,sawa2019euclidean}, which generalize spherical designs to weighted multi-shell configurations and can be viewed as a special case of cubature formulas~\cite{bannai2010cubature}.

A simple class of multi-shell QCCs can be constructed from tight Euclidean designs in $\mathbb{R}^2$. Tight Euclidean designs on $p$ concentric circles exist for any degree $t$~\cite{bajnok2006euclidean} with $1{\le} p{\le}\lfloor(t+5)/4 \rfloor$. Specifically, with $m{=}t+3-2p$ denoting the number of points in each shell, the logical constellation can be defined as
\begin{equation}\label{eq:QCC8}
    \mathcal{C}_{\text{plane}} = \bigcup_{s=1}^{p} \bigcup_{j=1}^{m} \left\{ r_s \left( \cos \frac{(2j+s)\pi}{m}, \sin \frac{(2j+s)\pi}{m} \right) \right\},
\end{equation}
with weights (up to a normalization factor)
\begin{equation}\label{eq:plane_QCC_weights}
    |w_s|^2 \propto \begin{cases}
\dfrac{1}{r_1^m}, & s = 1, \\[10pt]
\dfrac{(-1)^s}{r_s^m} \prod_{\substack{2 \leq l \leq p \\ l \neq s}} \dfrac{r_1^2 - r_l^2}{r_s^2 - r_l^2}, & s \geq 2.
\end{cases}
\end{equation}
where $s$ denotes the index of the $s$-th shell. The constellations are unions of concentric regular 
$m$-gons (one on each circle), with constant weights on each circle and appropriate relative rotations. As a specific example, consider a QCC built from a two-shell tight Euclidean 7-design~\cite[Sec.~5]{BannaiBannai2009SphericalEuclidean} in $\mathbb{R}^2$. In this case, the logical constellation consists of two hexagons with different radii $r_1$ and $r_2$, and a relative $\pi/6$ rotation:
\begin{equation}
\begin{aligned}
\mathcal{C}_{\mathrm{hexagon}}^{(r_1)} &= r_1\left\{\left(\cos\frac{j\pi}{3}, \sin\frac{j\pi}{3}\right) \mid j\in\mathbb{Z}_6\right\}, \\
\mathcal{C}_{\mathrm{hexagon}}^{(r_2)} &= r_2\left\{\left(\cos\frac{(2j+1)\pi}{6}, \sin\frac{(2j+1)\pi}{6} \right) \mid j\in\mathbb{Z}_6\right\}.
\end{aligned}
\end{equation}
By setting $r_2/r_1=\tau$, the weights are given by
\begin{equation}
\label{eqn:concentric_QSC_weights_hex}
|w_{\boldsymbol{\alpha}}|^2 \propto \begin{cases} 1, &   \boldsymbol{\alpha}\in\mathcal{C}_{\mathrm{hexagon}}^{(r_1)}, \\[4pt] \left(1/\tau\right)^6, & \boldsymbol{\alpha}\in\mathcal{C}_{\mathrm{hexagon}}^{(r_2)}. \end{cases} 
\end{equation} 
When $r_2{=}r_1{=}1$, this results in a QSC whose logical constellation is a 12-gon with parameters $(\!(2,2,d_E\approx0.07,\langle 12,12,12\rangle)\!)$, while for $\tau{=}2$, this yields a QCC with code parameters $(\!(2,2,d_E\approx0.26,\langle 7,8,18\rangle)\!)$, which achieves a larger resolution than the corresponding QSC that
confines all coherent states to a single energy level. This comes at the cost of reduced photon-loss correction capability (smaller $t_\downarrow$), although the number of detectable photon-loss events $(d_\updownarrow - 1)$ is increased.

Similarly, the same number of coherent states can be distributed across 3 spheres, with logical constellation
\begin{equation}
\mathcal C_{\text{square}}^{(r_s)}
 = \Bigl\{\, r_s\bigl(\cos\tfrac{(2j+s)\pi}{4},
                      \sin\tfrac{(2j+s)\pi}{4}\bigr)
      \,\big|\, j\in\mathbb Z_4 \Bigr\},
      \quad s=1,2,3
\end{equation}
associated with shell weights given by Eq.~\eqref{eq:plane_QCC_weights}. Choosing $r_1/r_2{=}1/2$ and $r_2/r_3{=}2/3$ yields a 3-shell QCC with parameters $(\!(2,2,d_E\approx0.44,\langle 6,8,12\rangle)\!)$. The resolution $d_E$ is further enhanced by distributing points across 3 shells, at the cost of reduced excitation-changing error protection.

For the multi-mode case, a $2$-mode multi-shell QCC can be constructed from the union \begin{equation} \label{eqn:8-16_cell_constellation}\mathcal{C}_{\mathrm{logical}} = \mathcal{C}_{\mathrm{8\text{-}cell}}^{(r_1)} \cup \mathcal{C}_{\mathrm{16\text{-}cell}}^{(r_2)}, \end{equation} where we use the notation $\mathcal{C}^{(r)} := r\,\mathcal{C}$ to indicate the two layers are composed of coherent states with amplitudes $r_1$ and $r_2$, respectively. When $r_1{=}r_2$, the logical constellation reduces to the single-shell constellation of Eq.~\eqref{eq:weighted_QSC} with uniform weights. By selecting an appropriate rotation such that the union of two logical constellations forms a 288-cell, this yields a $(\!(2,2,2{-}\sqrt{2},\langle 5,6,12\rangle)\!)$ QCC qubit, as mentioned above. In the case where $r_1{\neq}r_2$, choosing the weight ratio \begin{equation} \frac{w_{\mathrm{16\text{-}cell}}}{w_{\mathrm{8\text{-}cell}}}=\left(\frac{r_1}{r_2}\right)^4 
\end{equation} 
ensures that $\mathcal{C}_{\mathrm{logical}}$ forms a Euclidean 5-design~\cite[Prop.~16]{bajnok2007orbits}. By applying the same rotation to the second logical constellation as the single-shell case, this yields a $(\!(2,2,d_E{=}(2-\sqrt{2})\min\{r_1^2,r_2^2\},\langle 5,6,12\rangle)\!)$ multi-shell QCC with a superposition of coherent states with different amplitudes distributed on two concentric shells: an 8-cell and a 16-cell, as depicted in Fig.~\ref{fig:QCC}. The resolution is determined by the smaller of the two shell radii. Consequently, when restricting the energy to $\bar{n}_{\mathrm{QCC}}{=}1$ and $r_1{\neq}r_2$, we have $d_E{<}2-\sqrt{2}$, indicating that under the current logical constellation orientation (i.e., the rotation applied to the second logical constellation), the multi-shell code constellation achieves lower resolution than the case where all points reside on a single shell. However, as we demonstrate later in the pure-loss channel benchmark section, this multi-shell structure nevertheless offers performance enhancement over the single-shell structure across a wide range of pure-loss rates.

\begin{figure}[tb!]
    \centering
\includegraphics[width=0.9\columnwidth]{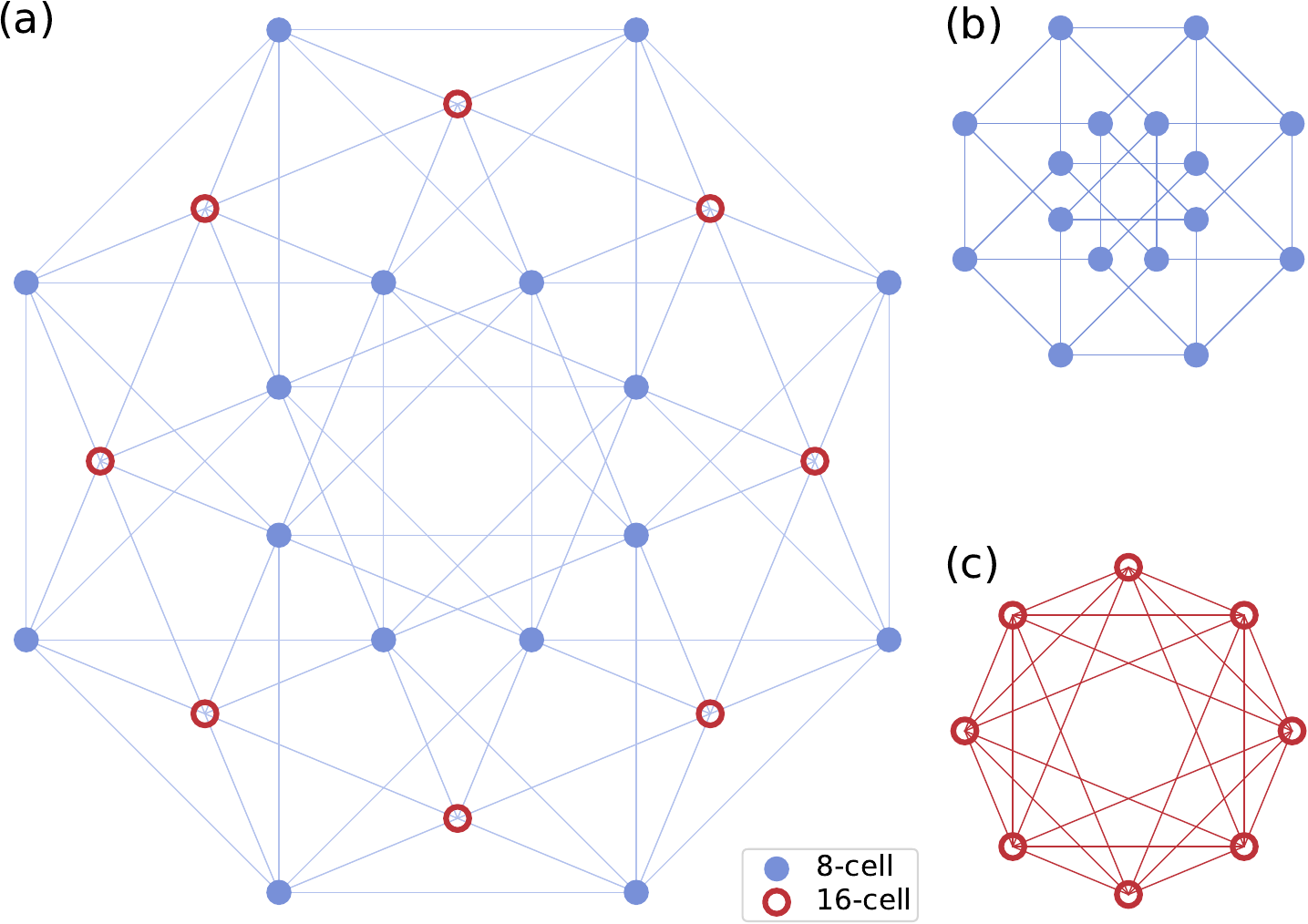}
    \caption{(a) Coherent-state logical constellation of a QCC constructed from the weighted union of (b) 8-cell and (c) 16-cell polytopes (see Eq.~\eqref{eqn:8-16_cell_constellation}). The 8-cell (blue) and 16-cell (red) components occupy coherent states with different mean photon numbers (energy). All panels are orthogonal projections onto the $B_4$ Coxeter plane.}
    \label{fig:QCC}
\end{figure} 

A larger QCC can be constructed by exploiting tight Euclidean 7-designs formed from unions of constellations (see \cite[Sec.~5]{BannaiBannai2009SphericalEuclidean} and \cite[Sec.~8.5, Ex.~7.1]{stroud1971approximate}). The logical constellation is distributed on two shells, with each shell containing 24 points that form a 24-cell (composed of a 16-cell and an 8-cell):
\begin{equation}
\mathcal{C}_{\mathrm{logical}} = \mathcal{C}_{\mathrm{24\text{-}cell}}^{(r_1)} \cup  \mathcal{C}_{\mathrm{24\text{-}cell}}^{(r_2)}.
\end{equation}
By setting the ratio $r_2/r_1 =\tau $, the weights of each constellation are
\begin{equation}\label{eqn:concentric_QSC_weights}
|w_{\boldsymbol{\alpha}}|^2 \propto \begin{cases} 1, & \boldsymbol{\alpha}\in\mathcal{C}_{\mathrm{24\text{-}cell}}^{(r_1)}, \\[4pt] (1/\tau)^6, & \boldsymbol{\alpha}\in\mathcal{C}_{\mathrm{24\text{-}cell}}^{(r_2)}, \end{cases} 
\end{equation} 
Notably, the optimal rotation in $\mathrm{SO}(4)$ that maximizes the resolution $d_E$ depends on the radius ratio $\tau$. When $1{<}\tau{<}1+\sqrt{2-\sqrt{2}}$, the bottleneck distance $d_E$ is determined by inner--outer shell pairs, and the optimal rotation angle depends on $\tau$. Once $\tau \ge 1 + \sqrt{2-\sqrt{2}}$, the bottleneck transitions to inner-shell pairs and $d_E$ saturates at $d_E = (2-\sqrt{2}) r_1^2$, independent of $\tau$. Under the normalized energy constraint, the resolution satisfies $d_E < 2-\sqrt{2}$. For instance, choosing $\tau=2$ yields a tight QCC with parameters $(\!(2,2,d_E\approx0.56,\langle 6,8,12\rangle)\!)$, whose resolution substantially exceeds that of the $r_1=r_2$ case ($d_E\approx0.3$, where two 288-cells lie within a single sphere). Thus, at fixed energy, the multi-shell QCC achieves superior resolution compared to the corresponding single-shell QSC. Moreover, when all points align on a single shell, the constellation forms a non-tight QSC based on a non-tight spherical 7-design, whereas the $r_1\ne r_2$ configuration yields a tight QCC grounded in a tight Euclidean 7-design.

\section{Symmetries and stabilizers}

The coherent-state constellations underlying QCCs induce natural symmetry constraints on the code space. In the single-shell, uniform-weight setting of QSCs, these constraints can be organized into Z-type and X-type stabilizers, and the logical code space can be characterized as the subspace that is invariant under both families of operators within the finite-dimensional span of the coherent-state constellation~\cite{denys20232, jain2024quantum}. The same stabilizer viewpoint extends to general QCCs and provides a convenient route to dissipative realizations via Lindblad dynamics.

\textbf{Z-type stabilizers.}
Z-type stabilizers restrict the support of QCC codewords to the specified constellation.
Let $\mathcal{S} = \bigcup_k \mathcal{V}_k$ denote the union of all constellation points. Choose a finite family of $n$-variable polynomials 
$\{P_i(\boldsymbol{\alpha})\}_{i}$ whose common zero set coincides
with $\mathcal S$, and denote
\begin{equation}
P_i(\boldsymbol{\alpha})
 = \sum_{\boldsymbol{u}} c^{(i)}_{\boldsymbol{u}}
   \prod_{j=1}^n \alpha_j^{u_j},
\end{equation}
where $\boldsymbol{u}=(u_1,\ldots,u_n)$ is a multi-index.
Define the corresponding annihilation-operator polynomials
\begin{equation}
F_i := P_i(\hat{\boldsymbol{a}})
 = \sum_{\boldsymbol{u}} c^{(i)}_{\boldsymbol{u}}
   \prod_{j=1}^n a_j^{u_j}.
\end{equation}
From the coherent-state eigenvalue relation ($a_j|\boldsymbol{\alpha}\rangle
= \alpha_j|\boldsymbol{\alpha}\rangle$), it follows that
\begin{equation}
F_i|\boldsymbol{\alpha}\rangle
 = P_i(\boldsymbol{\alpha})|\boldsymbol{\alpha}\rangle
\end{equation}
for any coherent state $|\boldsymbol{\alpha}\rangle$. Every QCC codeword
\begin{equation}
|\mathcal{C}_k\rangle \propto
 \sum_{\boldsymbol{\alpha}\in\mathcal{V}_k}
 \sqrt{w^{(k)}_{\boldsymbol{\alpha}}}\,|\boldsymbol{\alpha}\rangle
\end{equation}
is supported on $\mathcal{S}$, so $P_i(\boldsymbol{\alpha})=0$ for all
$\boldsymbol{\alpha}\in\mathcal{S}$ and hence $F_i|\mathcal{C}_k\rangle=0$
for all $i$ and $k$. Consequently,
\begin{equation}
\mathcal{H}_{\mathcal{S}}
 := \mathrm{Span}\{\,|\boldsymbol{\alpha}\rangle :
                     \boldsymbol{\alpha}\in\mathcal{S}\,\}
 \;\subseteq\; \bigcap_i \ker F_i,
\label{eq:Zmanifold_support}
\end{equation}
so the finite-dimensional subspace $\mathcal{H}_{\mathcal{S}}$ is contained
in the joint kernel of all Z-type stabilizers and contains the QCC code
space.

A dissipative implementation of Z-type stabilizers uses the $F_i$ as jump operators in a Lindblad master equation,
\begin{equation}
\dot{\rho}
= \sum_i \mathcal{D}[F_i]\rho,
\qquad
\mathcal{D}[L]\rho := L\rho L^\dagger - \tfrac12\{L^\dagger L,\rho\},
\end{equation}
as in reservoir-engineering schemes for bosonic codes~\cite{lescanne_exponential_2020}.
Any density operator supported entirely on $\bigcap_i\ker F_i$ is left invariant by this dynamics, while components violating at least one polynomial constraint are damped.

\textbf{X-type stabilizers.}
X-type stabilizers are symmetries that preserve the relative phases and amplitudes within each logical constellation.
Let $S_X$ denote the group of passive linear-optics unitaries that permute each constellation:
\begin{equation}
S_X :=
\bigl\{\, U \in \mathrm{U}_{\mathrm{pass}}(n) : 
U(\mathcal{V}_k, \mathcal{W}_k) = (\mathcal{V}_k, \mathcal{W}_k)\ \ \forall\, k \,\bigr\},
\end{equation}
where $U(\mathcal{V}_k, \mathcal{W}_k) = (\mathcal{V}_k, \mathcal{W}_k)$ signifies that $U$ maps the point set $\mathcal{V}_k$ onto itself while preserving the associated weights $\mathcal{W}_k$ as a weighted multiset.
By construction, each $U \in S_X$ permutes the coherent states within every $\mathcal{V}_k$ without altering their weights, and therefore
\begin{equation}
U\,|\mathcal{C}_k\rangle = |\mathcal{C}_k\rangle \quad \forall\,U \in S_X,\ \forall\,k,
\end{equation}
as a consequence, including dissipators of the form $\mathcal{D}[U-\mathds{1}]$ enforces invariance under $S_X$ and drives the system toward states satisfying $U\rho U^\dagger = \rho$ for all $U\in S_X$.

\textbf{The Z-/X-type stabilized subspace.}
The combined action of the Z- and X-type stabilizers restricts the code
space to the subspace
\begin{equation}
\mathcal{M}_{ZX} :=
\Bigl(\bigcap_i \ker F_i\Bigr)
\;\cap\;
\Bigl(\bigcap_{U\in S_X} \ker(U-\mathds{1})\Bigr),
\end{equation}
which satisfies
\begin{equation}
\mathrm{Span}\{|\mathcal{C}_k\rangle\}_k
\;\subseteq\; \mathcal{M}_{ZX}.
\label{eq:code_in_MZX_final}
\end{equation}
For single-shell constellations with uniform weights (the QSC setting), the only
states in $\mathcal{H}_{\mathcal{S}}$ that are simultaneously annihilated
by all $F_i$ and invariant under all $U\in S_X$ are linear combinations
of the logical codewords. Equivalently,
\begin{equation}
\mathcal{M}_{ZX}\cap\mathcal{H}_{\mathcal{S}}
 = \mathrm{Span}\{|\mathcal{C}_k\rangle\}_k.
\end{equation}
Thus, within the finite-dimensional subspace spanned by the constellation
coherent states, Z-type and X-type stabilizers uniquely specify the QSC
code space and can be used to stabilize it dissipatively via the jump
operators $\{F_i\}$ and $\{U-\mathds{1}\}_{U\in S_X}$.

For general QCCs with non-uniform weights and multiple shells, the relation between stabilizers and the code space is more subtle. The Z-type constraints confine the coherent-state support to $\mathcal{S}$, and the X-type constraints enforce permutation symmetry within each $S_X$-orbit of $\mathcal{S}$ (for example, a fixed shell and weight class), but the relative amplitudes assigned to different orbits remain free parameters. For a given logical index $k$, the orbit-symmetric part of $\mathcal{M}_{ZX}$ therefore has dimension larger than one: it contains several independent symmetric combinations of coherent states supported on different shells. The logical state $|\mathcal{C}_k\rangle$ corresponds to a particular choice of amplitudes across these orbits, while other choices produce additional states that still satisfy all Z-type and X-type constraints and remain invariant under the corresponding dissipative dynamics. Consequently, whenever such residual amplitude degrees of freedom are present, one has
\begin{equation}
\mathrm{Span}\{|\mathcal{C}_k\rangle\}_k \;\subsetneq\; \mathcal{M}_{ZX},
\end{equation}
and additional constraints acting within $\mathcal{M}_{ZX}$ are required to uniquely specify the QCC code space. A complete characterization of such amplitude-profile constraints in the general multi-shell setting remains an open question for future investigation.

\section{Performance under pure-loss noise}
We benchmark QCC performance under bosonic pure-loss noise using the composite channel~\cite{albert2018performance}
\begin{equation}
\Lambda = \mathcal{R} \circ \mathcal{N}_\gamma \circ \mathcal{E},
\end{equation}
where $\mathcal{E}$ is the encoding map, $\mathcal{N}_\gamma$ is a pure-loss channel with loss probability $\gamma$, and $\mathcal{R}$ is the recovery operation.
For a single bosonic mode, the pure-loss channel admits the Kraus representation
\begin{equation}
\mathcal{N}_\gamma(\rho) = \sum_{\ell=0}^{\infty} E_\ell \rho E_\ell^\dagger, \label{eq:pureloss_eta}
\end{equation}
with Kraus operators~\cite{nielsen2010quantum}
\begin{equation}
E_\ell = \left(\frac{\gamma}{1-\gamma}\right)^{\ell/2} \frac{\hat{a}^\ell}{\sqrt{\ell!}} (1-\gamma)^{\hat{n}/2}, \label{eq:pureloss_gamma_alt}
\end{equation}
where $\hat{n}=\hat{a}^\dagger \hat{a}$ is the photon-number operator, and the multimode case follows by taking tensor products over modes.

\begin{figure*}[tbp]
    \centering
    \includegraphics[width=\textwidth]{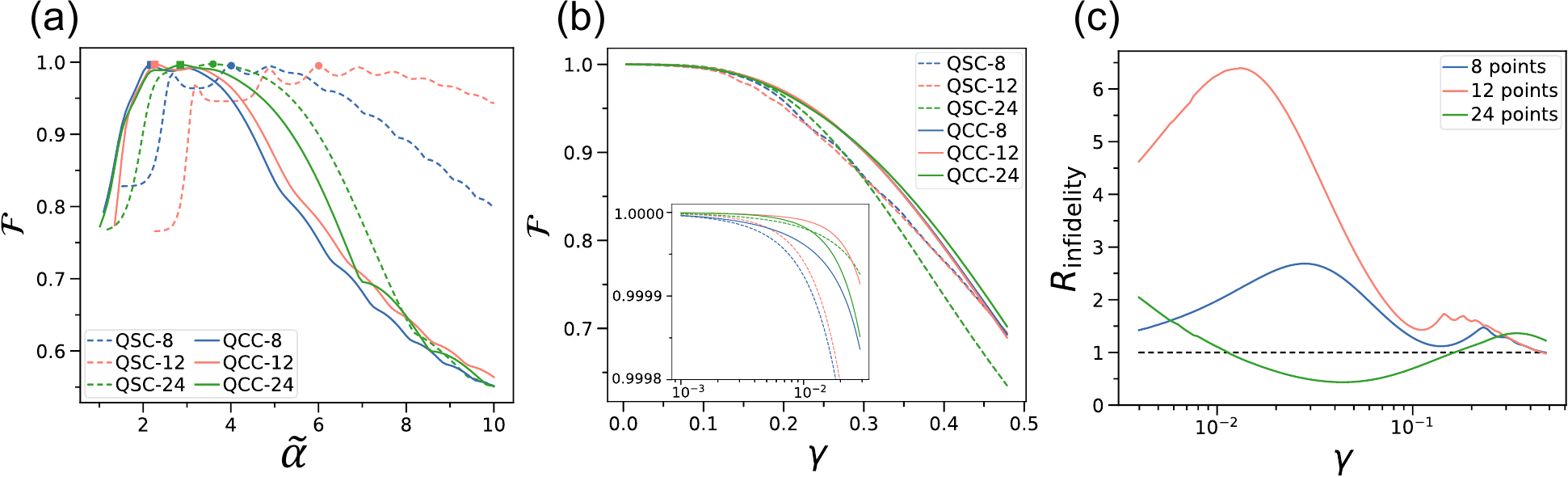}
\caption{(a) Entanglement fidelity $\mathcal{F}$ of QSCs and QCCs as a function of coherent-state amplitude scale $\tilde{\alpha}$ for fixed pure-loss rate $\gamma = 0.1$. Markers indicate the optimal $\tilde{\alpha}_{\mathrm{op}}$ for each code. (b) Entanglement fidelity $\mathcal{F}$ as a function of the pure-loss rate $\gamma$, evaluated at these optimized amplitudes (inset: zoom for small $\gamma$). (c) Relative improvement of entanglement infidelity $R_{\mathrm{infidelity}} = (1-\mathcal{F})/(1-\mathcal{F}')$ of QCCs (primed) over QSCs (unprimed) with the same number of constellation points; the black dashed line $R_{\mathrm{infidelity}} = 1$ denotes equal performance. In all panels, blue, orange, and green curves correspond to constellations with $8$, $12$, and $24$ points, respectively.}
    \label{fig:pure-loss_fid}
\end{figure*}

To quantify the performance, we employ the entanglement fidelity $\mathcal{F}$ as our benchmark metric, following~\cite{albert2018performance,jain2024quantum,denys20232}, given by
\begin{equation}
\mathcal{F} = \langle \Phi_K| (\openone \otimes \Lambda)(|\Phi_K\rangle\langle\Phi_K|) |\Phi_K\rangle \;, \label{eq:Fe_def} 
\end{equation}
where $|\Phi_K\rangle$ is the maximally-entangled state shared between a logical subsystem $A$ of dimension $K$ and an ancillary reference system $B$ of the same dimension: 
\begin{equation} 
|\Phi_K\rangle = \frac{1}{\sqrt{K}} \sum_{m=0}^{K-1} |m_A m_B\rangle. 
\end{equation}
This quantity measures the overlap between the ideal maximally entangled state and the state obtained after the channel $\Lambda$ acts on subsystem $A$, thereby quantifying code robustness under pure photon loss.

To compare QSCs and our newly constructed QCCs in a fair manner, we proceed as follows: 
As discussed in the examples, many cubature formulas admit families of realizations in which the radii of the supporting shells can be varied while preserving a fixed cubature degree. When the constellation points are distributed over several concentric spheres with distinct radii, we obtain a multi-shell configuration that defines a QCC. The same cubature formulas also admit specializations in which all points lie on a single sphere; in this case the cubature reduces to a spherical design and the resulting codes are QSCs. In our comparison, we fix the total number of coherent states in each logical constellation, so that the only distinction between the two families is whether the points occupy a single sphere or multiple concentric spheres. This provides a natural framework for benchmarking QSCs against their multi-shell QCC counterparts. 

In particular, we examine three pairs of logical qubit codes, denoted QCC-$\ell$ and QSC-$\ell$ for $\ell \in \{8,12,24\}$, where $\ell$ represents the number of coherent states per logical constellation. The QSC-8 and QSC-12 codes are single-mode $\ell$-gon cat qubits with logical constellations forming regular $\ell$-gons, while QSC-24 has a logical constellation given by the vertices of the 24-cell embedded in $\mathbb{C}^2$. The corresponding QCC-8 is a single-mode, two-shell code constructed from a tight Euclidean 5-design in $\mathbb{R}^2$ (with $r_1/r_2=1/2$), and QCC-12 is a single-mode, three-shell code constructed from a tight Euclidean 7-design in $\mathbb{R}^2$ (with $r_1/r_2=1/2$ and $r_2/r_3=2/3$), both specified in Eq.~\eqref{eq:QCC8}. The QCC-24 code, given in Eq.~\eqref{eqn:8-16_cell_constellation}, employs a two-shell configuration in which a 16-cell and an 8-cell are placed on distinct radii with $r_{\text{16-cell}}/r_{\text{8-cell}}{=}2$, constituting a two-mode code.

The entanglement fidelity $\mathcal{F}$ depends sensitively on the energy of the logical states. For a given logical constellation $\mathcal{C}_k$, we introduce an overall amplitude scale $\tilde{\alpha}>0$ and consider the rescaled constellation $\tilde{\mathcal{C}}_k = \tilde{\alpha}\,\mathcal{C}_k$. Varying $\tilde{\alpha}$ changes the mean photon number of the code and therefore its entanglement fidelity. To compare QSCs and QCCs, we fix the pure-loss rate to $\gamma = 0.1$ and, for each code, sweep $\tilde{\alpha}$ over a broad range. The resulting entanglement fidelity $\mathcal{F}$ as a function of the amplitude $|\alpha|$ is shown in Fig.~\ref{fig:pure-loss_fid}(a) for the same set of QSCs and QCCs. For each code, we identify the amplitude scale $\tilde{\alpha}_{\text{op}}$ that maximizes $\mathcal{F}$ and find that every QCC reaches its optimal fidelity at a smaller $\tilde{\alpha}_{\text{op}}$ than its QSC counterpart with the same number of constellation points. In other words, QCCs require less energy to achieve their optimal entanglement fidelity than the corresponding QSCs.

Using the optimal scale $\tilde{\alpha}_{\mathrm{op}}$ for each code, Fig.~\ref{fig:pure-loss_fid}(b) shows the entanglement fidelity $\mathcal{F}$ as a function of the pure-loss rate $\gamma$ for the different QSCs and QCCs. Over a broad interval of $\gamma$, every QCC (constructed from a degree-$t$ Euclidean design distributed over concentric shells) attains a markedly higher entanglement fidelity than its QSC counterpart with the same number of constellation points.

To further quantify the performance of QCCs relative to their QSC counterparts, in Fig.~\ref{fig:pure-loss_fid}(c) we plot the relative entanglement infidelity $R_\text{infidelity}=(1-\mathcal{F})/(1-\mathcal{F}')$ as a function of $\gamma$, where $\mathcal{F}$ and $\mathcal{F}'$ are the entanglement fidelities of the QSC and QCC with the same number of constellation points, evaluated at amplitudes optimized at $\gamma=0.1$. A value $R_\text{infidelity}=1$ indicates equal performance, while $R_\text{infidelity}>1$ (respectively $R_\text{infidelity}<1$) means that the QCC performs better (respectively worse) than its QSC counterpart. We find $R_\text{infidelity}>1$ over a wide range of $\gamma$, indicating that in these examples the QCCs generally outperform the QSCs. Interestingly, as discussed in Sec.~\ref{sec:examples}, even though QCC-24 has a smaller resolution $d_E$ than QSC-24 and both codes share the same loss-protection order $\langle 5,6,12\rangle$, the two-mode multi-shell QCC-24 still yields a higher entanglement fidelity than QSC-24 over most of the $\gamma$ range (green curves); QSC-24 only slightly surpasses QCC-24 in a narrow window at small $\gamma$ (see the inset).

 \section{Discussion}
We have introduced Quantum Cubature Codes (QCCs), a comprehensive and mathematically rigorous  framework for constructing bosonic codes based on superpositions of coherent states. By connecting the requirements of quantum error correction to the theory of cubature formulas~\cite{cools1992survey,cools2003encyclopaedia,stroud1971approximate,sawa2019euclidean}, our formalism provides a systematic methodology for designing codes that satisfy the Knill-Laflamme condition~\cite{knill1997theory}.

The QCC framework offers a unifying perspective, revealing that established codes such as cat codes and QSCs are specialized cases corresponding to uniform superposition of coherent states on a single energy shell. Crucially, QCCs expand the design space significantly, enabling the systematic exploration of codes with non-uniform superpositions and multi-shell constellations. We leveraged this formalism to discover several novel codes derived from Euclidean designs~\cite{bannai2015survey,sawa2019euclidean}. Numerical benchmarking demonstrates that these multi-shell QCCs can outperform their single-shell counterparts. The key insight is that utilizing the radial degrees of freedom in the phase space allows for enhanced geometric separation between logical states at fixed energy, leading to superior error correction performance.

The introduction of QCCs opens several promising avenues for future research. The design space enabled by cubature formulas is vast. Further investigation into different classes of cubature formulas and tight Euclidean designs may yield codes optimized for specific hardware constraints and biased noise models~\cite{li2023correctingbiasednoiseusing,hanggli2020enhanced,Stafford_2023,Zhang_2023}. 
Furthermore, investigating how QCCs can be integrated into a larger fault-tolerant architecture, including concatenation with outer codes~\cite{Berent_2024,Chamberland_2022,guillaud2023quantum,Fukui_2023,le2023high,xu2023autonomous,Zhang_2023} and the implementation of logical gates, is essential for their long-term applicability. 
Exploring how cubature formulas can be used to obtain codes with superposition of squeezed states~\cite{Korolev_2024,gutman2025squeezedvacuumbosoniccodes} is also interesting.
The QCC formalism provides the necessary mathematical tools to harness the full potential of the bosonic phase space, paving the way for the development of more efficient and robust quantum technologies.

 \section*{Acknowledgments}
We thank R. Cools for providing access to the online encyclopedia of cubature formulas.
AT is supported by the CQT PhD scholarship, the Google PhD fellowship program, and the CQT Young Researcher Career Development Grant.
KB thanks HQCC WP 2.0 for financial support.

\bibliography{QCC}

\clearpage
\onecolumngrid
\pagebreak
\widetext

\begin{center}
\textbf{\large Supplemental Materials: Quantum Cubature Codes}

\vspace{0.25cm}

Yaoling Yang$^{1}$,  Andrew Tanggara$^{2,3}$, Tobias Haug$^{4}$, Kishor Bharti$^{5,1}$

\vspace{0.25cm}

$^{1}${\small \em Quantum Innovation Centre (Q.InC), Agency for Science, Technology and Research (A*STAR), 2 Fusionopolis Way, Innovis
\#08-03, Singapore 138634, Republic of Singapore}\\

$^{2}${\small \em Centre for Quantum Technologies, National University of Singapore, 3 Science Drive 2, Singapore 117543}\\

$^{3}${\small \em Nanyang Quantum Hub, School of Physical and Mathematical Sciences, Nanyang Technological 
University, Singapore 639673}\\

$^{4}${\small \em Quantum Research Center, Technology Innovation Institute, Abu Dhabi, UAE}\\

$^{5}${\small \em Institute of High Performance Computing (IHPC), Agency for Science, Technology and Research (A*STAR), 1 Fusionopolis
Way, \#16-16 Connexis, Singapore 138632, Republic of Singapore}

\vspace{0.25cm}

\end{center}

\date{\today}
\setcounter{equation}{0}
\setcounter{figure}{0}
\setcounter{table}{0}

\setcounter{section}{0}  
\renewcommand{\thesection}{S\arabic{section}} 

\setcounter{theorem}{0}
\renewcommand{\thetheorem}{S\arabic{theorem}}

\makeatletter

\renewcommand{\theequation}{S\arabic{equation}}
\renewcommand{\thefigure}{S\arabic{figure}}

\section{Proof of Theorem~\ref{QCC_thm}: Knill-Laflamme Conditions for Quantum cubature Codes}
\label{SM:section2}

This section provides the proof of Theorem~\ref{QCC_thm}. We demonstrate that quantum cubature codes (QCCs) constructed from cubature formulas of degree $t$ can correct $\lfloor t/2 \rfloor$ photon loss errors and detect up to $t$ photon gain and loss errors. We will first discuss the mapping between the cubature points and the coherent amplitudes $\boldsymbol{\alpha}$. Then, we verify that these QCCs satisfy the Knill-Laflamme (KL) conditions in the asymptotic regime, where the energy is sufficiently large such that distinct coherent states become approximately orthogonal. 

Specifically, we consider QCCs subject to noise-induced excitation changes in each mode, which are described by general ladder operators
\begin{equation} \label{SM:ladder_operator}
L_{\mathbf{p},\mathbf{q}}(\boldsymbol{a}^\dagger,\boldsymbol{a})= \prod_{j=1}^n (a_j^\dagger)^{p_j} a_j^{q_j}, 
\end{equation}
where $a_j^\dagger$ and $a_j$ are the creation and annihilation operators for the $j$-th mode, and $\mathbf{p},\mathbf{q}\in\mathbb{N}^{n}$ specify the number of photon creation and annihilation events in each of the $n$ modes, respectively. We analyze both the KL correction condition and the KL detection condition.

\subsection{Mapping cubature points to coherent amplitudes}
Cubature formulas are primarily studied for the real domain $\mathbb{R}^n$, where the cubature points are taken from Euclidean space. However, an $n$-mode coherent state $\ket{\boldsymbol{\alpha}}$ takes values in the complex space $\mathbb{C}^n$. Therefore, a mapping between real and complex spaces is needed to specify points in the complex domain while preserving the cubature formula property in Eq.~\eqref{eq:cubature} in the main text.

We mainly focus on two mapping methods. The simplest mapping method is natural embedding. One can directly identify $\mathbb{R}^d$ with the real subspace of $\mathbb{C}^d$; this naturally maintains the cubature formula property. When the real coordinates come from an even-dimensional space $\mathbb{R}^{2d}$, a standard mapping method pairs consecutive real coordinates $(x_1,x_2,\ldots,x_{2d})$ as 
\begin{equation}
(x_{2j-1}, x_{2j}) \mapsto x_{2j-1} + ix_{2j}
\end{equation}
to yield points in $\mathbb{C}^d$. To show that the coordinates after mapping still form a cubature formula that approximates certain integrals with a certain degree of polynomial precision, we introduce the following lemma:

\begin{lemma}[Complex extension of cubature formulas] \label{SM:Complexification_cubature}
Let $(\mathcal V,\mathcal W)$ be a cubature formula of degree $t$ on $\Omega\subset\mathbb R^{n}$ with measure $\sigma$, such that 
\begin{equation}  
\sum_{x\in\mathcal V} w_x\, g(x) \;=\; \int_{\Omega} g(x)\, d\sigma(x)
\quad\text{for all } g\in \mathscr{P}_t(\mathbb R^{n}),
\end{equation}
where $\mathscr{P}_t(\mathbb R^{n})$ denotes real-coefficient polynomials on $\mathbb R^{n}$ of total degree $\le t$.
Then the same equality holds for every complex-coefficient polynomial $h$ of total degree $\le t$
, namely
\begin{equation}
\sum_{x\in\mathcal V} w_x\, h(x) \;=\; \int_{\Omega} h(x)\, d\sigma(x).
\end{equation}
\end{lemma}

\begin{proof}
Write $h(x)=u(x)+iv(x)$ with $u(x),v(x)\in \mathscr{P}_t(\mathbb R^{n})$ real-coefficient polynomials.
By linearity of the sum and the integral,
\begin{equation}
\sum_{x\in\mathcal V} w_x\, h(x)
= \sum_{x\in\mathcal V} w_x\, u(x) + i \sum_{x\in\mathcal V} w_x\, v(x)
= \int_{\Omega} u(x)\, d\sigma(x) + i \int_{\Omega} v(x)\, d\sigma(x)
= \int_{\Omega} h(x)\, d\sigma(x).
\end{equation}
\end{proof}

Therefore, by using the standard mapping method from even-dimensional $\mathbb{R}^{2D}$ to $\mathbb{C}^{D}$, the resulting polynomial $f(\boldsymbol{\alpha})$ with $\boldsymbol{\alpha} \in \mathbb{C}^{D}$ is a complex-coefficient polynomial with the same degree $t$, and thus the cubature formula relation holds.

\subsection{Verifying the KL Correction Conditions}

 Consider a general QCC whose logical codeword are constructed from $K$ cubature formulas. The $k$-th cubature formula is specified by a point set $\mathcal{V}_k:=\{\boldsymbol{\alpha}_i\}_i$ and corresponding weights $\mathcal{W}_k:=\{w_i^{(k)}\}_i$. These cubature formulas compute identical integration values over the same domain $\Omega$ with dicrete weighted summation over distinct set of points $\mathcal{V}_k$, namely 
\begin{equation}
\sum_{\boldsymbol{\alpha} \in \mathcal{V}_1} w_{\boldsymbol{\alpha}}^{(1)} f(\boldsymbol{\alpha}) =  \sum_{\boldsymbol{\alpha} \in \mathcal{V}_2} w_{\boldsymbol{\alpha}}^{(2)} f(\boldsymbol{\alpha}) = \cdots =\int_{\Omega} f(\boldsymbol{\alpha})\mathrm{d}\sigma(\boldsymbol{\alpha}) = \mathrm{constant},
\end{equation}
holds for all polynomials $f$ of degree at most $t$, where $w^{(k)} \in \mathcal{V}_k$ is the associated weights of the $k$-th cubature formula with the subscript $(k)$.  The $k$-th logical codeword $\ket{\mathcal{C}_k}$ is defined as 
\begin{equation}\label{SM:QCC}
\ket{\mathcal{C}_k} \propto \sum_{\boldsymbol{\alpha} \in \mathcal{V}_k} \sqrt{w_{\boldsymbol{\alpha}}^{(k)}} \,|\boldsymbol{\alpha}\rangle,
\end{equation}
where $w_{\boldsymbol{\alpha}}^{(k)} \in \mathcal{W}_k$, and $\boldsymbol{\alpha} \in \mathcal{V}_k$. Note that we use the proportionality symbol rather than equality since the logical states defined above are properly normalized only in the large energy limit where $\|\boldsymbol{\alpha}\| \rightarrow \infty$. QCCs can correct loss errors or gain errors separately; however, they are not generally guaranteed to correct simultaneous loss and gain errors. For pure loss errors, the ladder operators $L_{\mathbf{p},\mathbf{q}}(\boldsymbol{a}^\dagger,\boldsymbol{a})$ in Eq.~\eqref{SM:ladder_operator} become: 
\begin{equation} \label{ladder_loss_operators}
L_{\mathbf{0},\mathbf{q}}(\boldsymbol{a}^\dagger,\boldsymbol{a})= \prod_{j=1}^n  a_j^{q_j}, 
\end{equation} 
where $\mathbf{p}=0$ indicates there are no gain errors and $\mathbf{q}= (q_1,q_2,\ldots, q_n)$ denote the number of loss events on each mode.

To satisfy the KL correction condition, we require
\begin{equation}\label{SM:eq:kl-correction}
\langle \mathcal{C}_k \mid E_{\mu}^\dagger E_{\nu} \mid \mathcal{C}_\ell \rangle = \delta_{k\ell} c_{\mu\nu},
\end{equation}
where $c_{\mu\nu}$ are constants independent of the logical state indices $k, \ell$, $\delta_{k\ell}$ is the Kronecker delta, and $E_\mu, E_\nu$ are taken from the error set $\{L_{\mathbf{0},\mathbf{q}}(\boldsymbol{a}^\dagger,\boldsymbol{a})\}$, where the error operators are associated with different $\mathbf{q}$. Setting $E_\mu = L_{\mathbf{0},\mathbf{q}^{(\mu)}}$ and $E_\nu = L_{\mathbf{0},\mathbf{q}^{(\nu)}}$, we have
\begin{equation}\label{KL_EE}
E_{\mu}^\dagger E_{\nu} = (L_{\mathbf{0},\mathbf{q}^{(\mu)}})^\dagger L_{\mathbf{0},\mathbf{q}^{(\nu)}} = \prod_{j=1}^n (a_j^\dagger)^{q_{j}^{(\mu)}} a_j^{q_j^{(\nu)}}.
\end{equation}
Substituting Eqs.~\eqref{KL_EE} and~\eqref{SM:QCC} into Eq.~\eqref{eq:kl-correction} yields
\begin{equation} \label{SM:KL_correction_off_diagonal}
\bra{\mathcal{C}_k} E_{\mu}^\dagger E_{\nu} \ket{\mathcal{C}_\ell} = \bra{\mathcal{C}_k} \prod_{j=1}^n (a_j^\dagger)^{q_j^{(\mu)}} a_j^{q_j^{(\nu)}} \ket{\mathcal{C}_\ell} \propto \sum_{\substack{\boldsymbol{\alpha} \in \mathcal{C}_k \\ \boldsymbol{\beta} \in \mathcal{C}_\ell}} \sqrt{w_{\boldsymbol\alpha}^{(k)} w_{\boldsymbol\beta}^{(\ell)}} f_{(\mathbf{q}^{(\mu)},\mathbf{q}^{(\nu)})}(\boldsymbol{\alpha}^*,\boldsymbol{\beta}) \langle{\boldsymbol{\alpha}}|\boldsymbol{\beta}\rangle,
\end{equation}
where $f_{(\mathbf{q}^{(\mu)},\mathbf{q}^{(\nu)})}(\boldsymbol{\alpha}^*,\boldsymbol{\beta}) = \prod_{j=1}^n (\alpha_j^*)^{q_j^{(\mu)}} \beta_j^{q_j^{(\nu)}}$. The inner product between coherent states is
\begin{equation}
|\langle{\boldsymbol{\alpha}}|\boldsymbol{\beta}\rangle| = \exp\Big(-\tfrac{1}{2}\|\boldsymbol{\alpha}-\boldsymbol{\beta}\|^2\Big),
\end{equation}
which depends on the Euclidean distance between constellation points.

In the large energy limit, $|\langle{\boldsymbol{\alpha}}|\boldsymbol{\beta}\rangle| \approx 0$ for $\boldsymbol{\alpha} \neq \boldsymbol{\beta}$. Therefore, for distinct logical codewords $\mathcal{C}_k \neq \mathcal{C}_\ell$, 
\begin{equation}\label{SM:KL_correction_diagonal}
    \bra{\mathcal{C}_k}\, E_{\mu}^\dagger E_{\nu} \,\ket{\mathcal{C}_\ell} \approx 0.
\end{equation}
This gives the satisfaction of the diagonal KL condition in Eq~\eqref{SM:eq:kl-correction} where $k\ne \ell$. 

For the diagonal KL condition where $k=\ell$, we have 
\begin{align}\label{SM:KL_diagonal}
\bra{\mathcal{C}_k} E_{\mu}^\dagger E_{\nu} \ket{\mathcal{C}_k} &\propto \sum_{\boldsymbol{\alpha}_x, \boldsymbol{\alpha}_y \in \mathcal{C}_k} \sqrt{w_{\boldsymbol{\alpha}_x}^{(k)} w_{\boldsymbol{\alpha}_y}^{(k)}} f_{(\mathbf{q}^{(\mu)},\mathbf{q}^{(\nu)})}(\boldsymbol{\alpha}_x^*,\boldsymbol{\alpha}_y) \langle{\boldsymbol{\alpha}_x}|\boldsymbol{\alpha}_y\rangle \notag \\
& \approx \sum_{\boldsymbol{\alpha} \in \mathcal{C}_k} w_{\boldsymbol{\alpha}}^{(k)} f_{(\mathbf{q}^{(\mu)},\mathbf{q}^{(\nu)})}(\boldsymbol{\alpha}^*,\boldsymbol{\alpha}).
\end{align}
where $f_{(\mathbf{q}^{(\mu)},\mathbf{q}^{(\nu)})}(\boldsymbol{\alpha}^*,\boldsymbol{\alpha}) = \prod_{j=1}^n (\alpha_j^*)^{q_j^{(\mu)}} \alpha_j^{q_j^{(\nu)}}$. The approximation in Eq.~\eqref{SM:KL_diagonal} also follows from the near-orthogonality of distinct coherent states in the large energy limit, which ensures $\langle{\boldsymbol{\alpha}_x}|\boldsymbol{\alpha}_y\rangle \approx \delta_{xy}$, eliminating cross terms. Thus, to satisfy the diagonal KL condition, we only need to ensure  
\begin{equation} \label{SM:KL_cubature}
     \sum_{\boldsymbol{\alpha} \in \mathcal{C}_k} w_{\boldsymbol{\alpha}}^{(k)} f_{(\mathbf{q}^{(\mu)},\mathbf{q}^{(\nu)})}(\boldsymbol{\alpha}^*,\boldsymbol{\alpha}) = \mathrm{constant}.
\end{equation}
Note that the left hand side of Eq.~\eqref{SM:KL_cubature} is a averaged monomial on a set of points $\{\boldsymbol{\alpha}\}$, as defined by the cubatue formula, this discrete summation approximate the integration up to certain polynomial degree $t$, namely
\begin{equation}  
 \sum_{\boldsymbol{\alpha} \in \mathcal{C}_k} w_{\boldsymbol{\alpha}}^{(k)} f_{(\mathbf{q}^{(\mu)},\mathbf{q}^{(\nu)})}(\boldsymbol{\alpha}^*,\boldsymbol{\alpha}) = \int_{\Omega} f_{(\mathbf{q}^{(\mu)},\mathbf{q}^{(\nu)})}(\boldsymbol{\alpha}^*,\boldsymbol{\alpha}) \, d\boldsymbol{\alpha}=\textrm{constant},
\quad\text{for all } f\in \mathscr{P}_t(\mathbb R^{n}),
\end{equation}
Thus, the cubature formulas provide a sufficient condition for satisfying Eq.~\eqref{SM:KL_cubature}. 

The polynomial $f_{(\mathbf{q}^{(\mu)},\mathbf{q}^{(\nu)})}$ has degree $|\mathbf{q}^{(\mu)}+\mathbf{q}^{(\nu)}|$. Since the error operators $E_\mu$ and $E_\nu$ are drawn from ${L_{\mathbf{0},\mathbf{q}}}$, the largest polynomial degree is $2|\mathbf{q}|_{\max}$. Therefore, the largest polynomial degree satisfying the KL condition is 
\begin{equation} 
|\mathbf{q}|_{\max}\le \lfloor t/2 \rfloor \le t/2. 
\end{equation}
Consequently, when combined with the KL diagonal condition in Eq.~\eqref{SM:KL_diagonal}, the QCC constructed from cubature formulas of degree $t$ forms a valid quantum error-correcting code and can correct up to $\lfloor t/2 \rfloor$ loss errors.

To see how QCCs handle the case where only gain errors happens
\begin{equation} \label{ladder_gain_operators}
L_{\mathbf{p},\mathbf{0}}(\boldsymbol{a}^\dagger,\boldsymbol{a})= \prod_{j=1}^n  (a_j^{\dagger}) ^{p_j}, 
\end{equation}  
where $\mathbf{p}=(p_1,p_2,...,p_n)$. By using the commutation relation for the $j$-th mode creation and annihilation operators~\cite{cahill1969ordered}, one obtains
\begin{equation}
a^{m}_j(a^\dagger_j)^{n}
=\sum_{k=0}^{\min(m,n)} \binom{m}{k}\binom{n}{k}\,k!\,(a^\dagger_j)^{\,n-k} a_j^{\,m-k}.
\end{equation}
This allows $ E_\mu^\dagger E_\nu$ with error operators $E_\mu$, $E_\nu$ taken from
$\{L_{\mathbf{p},\mathbf{0}}(\boldsymbol{a}^\dagger,\boldsymbol{a})\}$ to be expressed as 
\begin{equation}
   E_\mu^\dagger E_\nu = \prod_j a_j^{p_j^{(\mu)}}(a_j^{p_j^{(\nu)}})^\dagger = \sum_{k=0}^{\min(p_j^{(\mu)},p_j^{(\nu)})} \prod_j\binom{p_j^{(\mu)}}{k}\binom{p_j^{(\nu)}}{k}\,k!\,(a^\dagger_j)^{\,p_j^{(\nu)}-k} a_j^{\,p_j^{(\mu)}-k}.
\end{equation}
The KL diagonal condition in Eq.~\eqref{SM:KL_diagonal} now becomes:
\begin{equation}
\bra{\mathcal{C}_k} E_{\mu}^\dagger E_{\nu} \ket{\mathcal{C}_k}  \approx \sum_{\boldsymbol{\alpha} \in \mathcal{C}_k} w_{\boldsymbol{\alpha}}^{(k)} g_{(\mathbf{p}^{(\mu)},\mathbf{p}^{(\nu)})}(\boldsymbol{\alpha}^*,\boldsymbol{\alpha})
\end{equation}
where 
\begin{equation}
g_{(\mathbf{p}^{(\mu)},\mathbf{p}^{(\nu)})}(\boldsymbol{\alpha}^*,\boldsymbol{\alpha})
     = \sum_{k=0}^{\min(p_j^{(\mu)},p_j^{(\nu)})} \prod_j\binom{p_j^{(\mu)}}{k}\binom{p_j^{(\nu)}}{k}\,k!\,(\alpha^*_j)^{\,p_j^{(\nu)}-k} \alpha_j^{\,p_j^{(\mu)}-k}.
\end{equation}
Note that $g_{(\mathbf{p}^{(\mu)},\mathbf{q}^{(\nu)})}(\boldsymbol{\alpha}^*,\boldsymbol{\alpha})$ is a polynomial with the same highest degree as $f_{(\mathbf{p}^{(\mu)},\mathbf{p}^{(\nu)})}(\boldsymbol{\alpha}^*,\boldsymbol{\alpha})$, thus the same conclusion holds for the gain error cases. 

\subsection{Verifying the KL Detection Conditions}

The KL detection condition can be verified in a similar way, where now QCCs detect gain and loss error which happen at the same time. 
In this case, we need to ensure that for error operator set $\{E_\mu\}$,
\begin{equation}
\label{SM:eq_KL_detection}
\langle \mathcal{C}_k \mid E_\mu \mid \mathcal{C}_\ell \rangle = \delta_{k\ell} \lambda_\mu,
\end{equation}
where $\lambda_\mu$ is a constant. When the gain and loss errors happen simultaneously, the error $E_\mu$ is now described by the ladder operators:
\begin{equation}
L_{\mathbf{p},\mathbf{q}} = \prod_{j=1}^n (a_j^\dagger)^{q_j} a_j^{q_j}.
\end{equation}
Taking this into Eq.~\eqref{SM:eq_KL_detection}, one obtains:
\begin{equation}
\bra{\mathcal{C}_k} E_{\mu}\ket{\mathcal{C}_\ell} = \bra{\mathcal{C}_k} \prod_{j=1}^n (a_j^\dagger)^{p_j} a_j^{q_j} \ket{\mathcal{C}_\ell} \propto \sum_{\substack{\boldsymbol{\alpha} \in \mathcal{C}_k \\ \boldsymbol{\beta} \in \mathcal{C}_\ell}} \sqrt{w_{\boldsymbol\alpha}^{(k)} w_{\boldsymbol\beta}^{(\ell)}} f_{(\mathbf{p},\mathbf{q})}(\boldsymbol{\alpha}^*,\boldsymbol{\beta}) \langle{\boldsymbol{\alpha}}|\boldsymbol{\beta}\rangle,
\end{equation}
where $f_{(\mathbf{p},\mathbf{q})}(\boldsymbol{\alpha}^*,\boldsymbol{\beta}) = \prod_{j=1}^n (\alpha_j^*)^{p_j} \beta_j^{q_j}$. By the approximate orthogonality of different coherent states in the large-energy limit, the diagonal KL condition (for $k \neq \ell$) is satisfied. Moreover, by the properties of cubature formulas, the off-diagonal KL condition is also satisfied for $|\mathbf{p}+\mathbf{q}|\le t$, analogous to Eq.~\eqref{SM:KL_diagonal}. Consequently, the code can correct up to $\lfloor t/2 \rfloor$ gain and loss errors simultaneously.

\section{Discussion of Proposition~\ref{prop:const_size}: Bounds on Constellation Size}
\label{SM:section3}

In this section, we discuss bounds on the number of superposed coherent states per logical codeword in quantum cubature codes. Specifically, we show that there is a necessary lower bound arising from the underlying cubature formulas, as well as an upper bound demonstrating that quantum cubature codes can be constructed with at most a certain number of superposed coherent states per logical codeword.

\subsection{Lower bounds}
Since quantum cubature codes inherit their constellation structure from cubature formulas, classical results from approximation theory provide lower bounds on the number of cubature points, which translate into constraints on constellation size. For instance, the Fisher-type bound establishes a general lower bound for degree-$t$ cubature formulas on a domain $\Omega$: 
\begin{theorem*}[Fisher-type bound~\cite{stroud1971approximate,sawa2019euclidean}]\label{SM:thm_fisher}
    The number of points $\mathcal{N}$ for any degree-$t$ cubature formula satisfies
\begin{equation}\label{SM:fisher_type_bound} \mathcal{N}\ge\dim \mathscr{P}_{\lfloor t/2\rfloor}(\Omega) 
\end{equation} 
where $\mathscr{P}_{\lfloor t/2 \rfloor}(\Omega)$ denotes the polynomials space on $\Omega$ with degree at most $\lfloor t/2 \rfloor$.
\end{theorem*}
For centrally symmetric domains $\Omega$ where all odd-degree polynomial integrals vanish, the Fisher-type bound can be strengthened to the M\"{o}ller bound:
\begin{theorem*}[M\"{o}ller bound~\cite{moller1979lower,sawa2019euclidean}]\label{SM:thm_moller}
Let $t=2e + 1$ denote the degree of the cubature formula with centrally symmetric integral, then
\begin{equation}\label{SM:moller}
\mathcal{N} \geq \begin{cases}
2 \dim \mathscr{P}_e^*(\Omega) - 1 & \text{if } e \equiv 0 \pmod{2} \text{ and } \text{0 is included in the cubature points}, \\
2 \dim \mathscr{P}_e^*(\Omega) & \text{otherwise}.
\end{cases}
\end{equation}
where $\mathscr{P}_{e}^*(\Omega) = \bigoplus_{i=0}^{\left\lfloor e/2 \right\rfloor} \mathrm{Hom}_{e - 2i}(\Omega)$ and $\mathrm{Hom}_{e - 2i}(\Omega)$ is space of homogenous polynomial of degree $(e - 2i)$ on $\Omega$. 
\end{theorem*}
Notably, for QCCs case, since the superposition of coherent states does not include vacuum state ($\boldsymbol{\alpha}=0$), thus the lower bound in Eq.~\eqref{SM:moller} reduce to
\begin{equation} \label{SM:moller_reduced}
    \mathcal{N} \ge 2 \dim \mathscr{P}_e^*(\Omega)
\end{equation}

By combining the Fisher-type bound with the M\"{o}ller bound, we establish a necessary lower bound on the number of cubature points in cubature formulas, which directly translates to a lower bound on the constellation size for logical codewords of QCCs. Consequently, to construct a QCC with a specified error-correcting capability, the number of superposed coherent states in each logical constellation must meet or exceed this analytical threshold. This bound serves as a benchmark for evaluating the resource efficiency of quantum cubature codes: constellations that approach this lower bound achieve the desired degree-$t$ error correction with minimal superposed coherent states per logical codeword, thereby maximizing resource efficiency.

As a concrete example, consider cubature formulas on the unit sphere $\mathbb{S}^{D-1}\subset\mathbb{R}^D$ with the rotationally invariant surface measure. In such case the cubature formulas reduce to spherical designs, and the QCC constructed from such designs are QSCs. For even degree $t=2e$, the Fisher-type bound in Eq.~\eqref{SM:fisher_type_bound} implies 
\begin{equation}
\mathcal N \;\ge\; \dim \mathscr P_e(\mathbb{S}^{D-1}) = \binom{D+e-1}{e} + \binom{D+e-2}{e-1}.
\end{equation} 
For odd degree $t=2e+1$, the M\"{o}ller’s bound in Eq.~\eqref{SM:moller_reduced} gives
\begin{equation}
\mathcal N \;\ge\; 2\,\dim \mathscr P_e^{*}(\mathbb{S}^{D-1})
  = 2\binom{D+e-1}{e}.
\end{equation}
Therefore, the lower bound $\mathcal{N}$ for spherical designs of degree $t$ are
\begin{equation}\label{SM:fisher-type}
    \mathcal{N} \geq \begin{cases}  
\binom{D+e-1}{e} + \binom{D+e-2}{e-1} & \text{for\, } t = 2e, \\
2\binom{D+e-1}{e} & \text{for\, } t = 2e+1.
\end{cases}
\end{equation}
This is also known as Fisher-type bound for spherical designs~\cite{bannai2009survey}. Notably, cubature formulas that saturate the lower bound is called minimal cubature formulas, or alternatively, tight spherical design.

\begin{table}[!h]
\centering
\caption{The currently known tight spherical $t$-designs on $\mathbb{S}^{D-1}$ with $D\ge3$.}
\label{tab:tight_design}
\begin{tabular}{ccccl}
\hline
design & dimension & vertices & source  \\
\hline
$1^\star$ & $D$ & 2 & antipodal pair  \\
\hline
$2^\star$ & $D$ & $D+1$ & regular simplex  \\
\hline
$3^\star$ & $D$ & $2D$ & cross-polytope ($D$-orthoplex) \\
\hline
4 & 6 & 27 & $E_6$ root system  \\
4 & 22 & 275 &  Hamming scheme $H(11,3)$\\
\hline
5 & 3 & 12 & icosahedron\\
5 & 7 & 56 & $E_7$ root system \\
5 & 23 & 552 & McLaughlin-type 5-design \\
\hline
7 & 8 & 240 & $E_8$ root system  \\
7 & 23 & 4600 &McLaughlin-type 7-design \\
\hline
$11^\star$ & 24 & 196560 & minimal vectors of the Leech lattice $\Lambda_{24}$ \\
\hline
\end{tabular}
\end{table}

The lower bound derived from the Fisher-type bound and M\"{o}ller bound is generally not tight, namely for different design strength $t$ and dimension $D$, the tight design does not always exist. Specifically, for $\mathbb{S}^1 \subset \mathbb{R}^2$, every regular polygon with $p$ vertices forms a tight spherical $(p-1)$-design. For dimension $D\ge3$, however, only a few classes of tight spherical designs exist~\cite{bannai2009survey,delsarte_spherical_1977}, as summarized in Table~\ref{tab:tight_design}. In this table, ``$\star$'' indicates that the classification is complete for the corresponding design strength and the example shown is unique. For $t=4,5,7$, the classification remains incomplete. 

\begin{table}[htbp!]
\centering
\caption{Examples of QSCs with logical constellations constructed from tight spherical designs.}
\begin{tabular}{cccc
                @{\hspace{1.5em}}c
                @{\hspace{1.5em}}c
                @{\hspace{1.3em}}c
                @{\hspace{1.3em}}c}
\hline
logical constellation 
  & \makecell{vertices\\[-0.6ex](per logical const.)}
   &code constellation 
  & \makecell{vertices\\[-0.6ex](per codeword)} 
  & $n$ & $K$ & $\langle t_{\downarrow}, d_{\updownarrow}, d_{\downarrow}\rangle$ & $d_E$ \\
\hline
line segment ($D=2$) & 2 & $2K$-gon & $2K$  & 1 & $K$ & $\langle 2,2,2\rangle$ & $4\sin^2\frac{\pi}{2K}$ \\
icosahedron ($D=2$) & 12  & 2-icosahedron & 24& 2 & 2& $\langle4,6,6\rangle$ & 0.211 \\
tetrahedron (3-simplex) ($D=3$) & 3 & dodecahedron &20 & 2 & 5 & 3 & 0.509  \\
octahedron (3-orthoplex) ($D=3$) & 6   & 5-octahedron &30 & 2 & 5& 4&0.382 \\
hyper-octahedron (4-orthoplex) ($D=4$) & 8  & 24-cell & 24 & 2 & 3 & $\langle2,4,4\rangle$ & 1.000\\
hyper-tetrahedron (4-simplex) ($D=4$) & 5  & hyper-dodecahedron & 600 & 2 &120 & 3 & 0.073 \\
\hline
\end{tabular}
\label{tab:tight_QSC}
\end{table}

Using tight spherical designs as logical constellations, one can construct QSCs that can be considered ``efficient'' in terms of the number of superposed coherent states in the logical constellation. In Table~\ref{tab:tight_QSC}, we summarize some QSCs whose logical states are constructed from tight spherical designs. These real polytopes are embedded in 
$n$-dimensional complex space to form QSCs via different mapping methods~\cite{jain2024quantum}.

\subsection{Existence upper bounds} We have established analytical lower bounds on the number of superposed coherent states required to construct quantum cubature codes that correct loss errors of a given degree. A natural complementary question arises: does there exist an upper bound? Specifically, can any quantum cubature code with a specified error correction capability be constructed using at most a finite number of coherent states? Remarkably, the existence of such an upper bound follows from a classical result in cubature theory known as the Tchakaloff theorem~\cite{tchakaloff1957formules}, which establishes the existence upper bound $\mathcal{N}\leq \dim \mathscr{P}_t(\Omega)$. 
This indicates that, any quantum cubature code with degree-$t$ error protection capability can be constructed using at most $\dim \mathscr{P}_t(\Omega)$ superposed coherent states per logical codeword.
\end{document}




%% file: main.bbl
\begin{thebibliography}{67}%
\makeatletter
\providecommand \@ifxundefined [1]{%
 \@ifx{#1\undefined}
}%
\providecommand \@ifnum [1]{%
 \ifnum #1\expandafter \@firstoftwo
 \else \expandafter \@secondoftwo
 \fi
}%
\providecommand \@ifx [1]{%
 \ifx #1\expandafter \@firstoftwo
 \else \expandafter \@secondoftwo
 \fi
}%
\providecommand \natexlab [1]{#1}%
\providecommand \enquote  [1]{``#1''}%
\providecommand \bibnamefont  [1]{#1}%
\providecommand \bibfnamefont [1]{#1}%
\providecommand \citenamefont [1]{#1}%
\providecommand \href@noop [0]{\@secondoftwo}%
\providecommand \href [0]{\begingroup \@sanitize@url \@href}%
\providecommand \@href[1]{\@@startlink{#1}\@@href}%
\providecommand \@@href[1]{\endgroup#1\@@endlink}%
\providecommand \@sanitize@url [0]{\catcode `\\12\catcode `\$12\catcode
  `\&12\catcode `\#12\catcode `\^12\catcode `\_12\catcode `\%12\relax}%
\providecommand \@@startlink[1]{}%
\providecommand \@@endlink[0]{}%
\providecommand \url  [0]{\begingroup\@sanitize@url \@url }%
\providecommand \@url [1]{\endgroup\@href {#1}{\urlprefix }}%
\providecommand \urlprefix  [0]{URL }%
\providecommand \Eprint [0]{\href }%
\providecommand \doibase [0]{https://doi.org/}%
\providecommand \selectlanguage [0]{\@gobble}%
\providecommand \bibinfo  [0]{\@secondoftwo}%
\providecommand \bibfield  [0]{\@secondoftwo}%
\providecommand \translation [1]{[#1]}%
\providecommand \BibitemOpen [0]{}%
\providecommand \bibitemStop [0]{}%
\providecommand \bibitemNoStop [0]{.\EOS\space}%
\providecommand \EOS [0]{\spacefactor3000\relax}%
\providecommand \BibitemShut  [1]{\csname bibitem#1\endcsname}%
\let\auto@bib@innerbib\@empty
\bibitem [{\citenamefont {Shor}(1995)}]{PhysRevA.52.R2493}%
  \BibitemOpen
  \bibfield  {author} {\bibinfo {author} {\bibfnamefont {P.~W.}\ \bibnamefont
  {Shor}},\ }\bibfield  {title} {\bibinfo {title} {Scheme for reducing
  decoherence in quantum computer memory},\ }\href
  {https://doi.org/10.1103/PhysRevA.52.R2493} {\bibfield  {journal} {\bibinfo
  {journal} {Phys. Rev. A}\ }\textbf {\bibinfo {volume} {52}},\ \bibinfo
  {pages} {R2493} (\bibinfo {year} {1995})}\BibitemShut {NoStop}%
\bibitem [{\citenamefont {Steane}(1996)}]{PhysRevLett.77.793}%
  \BibitemOpen
  \bibfield  {author} {\bibinfo {author} {\bibfnamefont {A.~M.}\ \bibnamefont
  {Steane}},\ }\bibfield  {title} {\bibinfo {title} {Error correcting codes in
  quantum theory},\ }\href {https://doi.org/10.1103/PhysRevLett.77.793}
  {\bibfield  {journal} {\bibinfo  {journal} {Phys. Rev. Lett.}\ }\textbf
  {\bibinfo {volume} {77}},\ \bibinfo {pages} {793} (\bibinfo {year}
  {1996})}\BibitemShut {NoStop}%
\bibitem [{\citenamefont {Aharonov}\ and\ \citenamefont
  {Ben-Or}(1997)}]{aharonov1997fault}%
  \BibitemOpen
  \bibfield  {author} {\bibinfo {author} {\bibfnamefont {D.}~\bibnamefont
  {Aharonov}}\ and\ \bibinfo {author} {\bibfnamefont {M.}~\bibnamefont
  {Ben-Or}},\ }\bibfield  {title} {\bibinfo {title} {Fault-tolerant quantum
  computation with constant error},\ }in\ \href@noop {} {\emph {\bibinfo
  {booktitle} {Proceedings of the twenty-ninth annual ACM symposium on Theory
  of computing}}}\ (\bibinfo {year} {1997})\ pp.\ \bibinfo {pages}
  {176--188}\BibitemShut {NoStop}%
\bibitem [{\citenamefont {Fowler}\ \emph {et~al.}(2012)\citenamefont {Fowler},
  \citenamefont {Mariantoni}, \citenamefont {Martinis},\ and\ \citenamefont
  {Cleland}}]{PhysRevA.86.032324}%
  \BibitemOpen
  \bibfield  {author} {\bibinfo {author} {\bibfnamefont {A.~G.}\ \bibnamefont
  {Fowler}}, \bibinfo {author} {\bibfnamefont {M.}~\bibnamefont {Mariantoni}},
  \bibinfo {author} {\bibfnamefont {J.~M.}\ \bibnamefont {Martinis}},\ and\
  \bibinfo {author} {\bibfnamefont {A.~N.}\ \bibnamefont {Cleland}},\
  }\bibfield  {title} {\bibinfo {title} {Surface codes: Towards practical
  large-scale quantum computation},\ }\href
  {https://doi.org/10.1103/PhysRevA.86.032324} {\bibfield  {journal} {\bibinfo
  {journal} {Phys. Rev. A}\ }\textbf {\bibinfo {volume} {86}},\ \bibinfo
  {pages} {032324} (\bibinfo {year} {2012})}\BibitemShut {NoStop}%
\bibitem [{\citenamefont {Terhal}(2015)}]{RevModPhys.87.307}%
  \BibitemOpen
  \bibfield  {author} {\bibinfo {author} {\bibfnamefont {B.~M.}\ \bibnamefont
  {Terhal}},\ }\bibfield  {title} {\bibinfo {title} {Quantum error correction
  for quantum memories},\ }\href {https://doi.org/10.1103/RevModPhys.87.307}
  {\bibfield  {journal} {\bibinfo  {journal} {Rev. Mod. Phys.}\ }\textbf
  {\bibinfo {volume} {87}},\ \bibinfo {pages} {307} (\bibinfo {year}
  {2015})}\BibitemShut {NoStop}%
\bibitem [{\citenamefont {Albert}(2022)}]{albert2022bosonic}%
  \BibitemOpen
  \bibfield  {author} {\bibinfo {author} {\bibfnamefont {V.~V.}\ \bibnamefont
  {Albert}},\ }\bibfield  {title} {\bibinfo {title} {Bosonic coding:
  introduction and use cases},\ }\href@noop {} {\bibfield  {journal} {\bibinfo
  {journal} {arXiv preprint arXiv:2211.05714}\ } (\bibinfo {year}
  {2022})}\BibitemShut {NoStop}%
\bibitem [{\citenamefont {Terhal}\ \emph {et~al.}(2020)\citenamefont {Terhal},
  \citenamefont {Conrad},\ and\ \citenamefont {Vuillot}}]{Terhal_2020}%
  \BibitemOpen
  \bibfield  {author} {\bibinfo {author} {\bibfnamefont {B.~M.}\ \bibnamefont
  {Terhal}}, \bibinfo {author} {\bibfnamefont {J.}~\bibnamefont {Conrad}},\
  and\ \bibinfo {author} {\bibfnamefont {C.}~\bibnamefont {Vuillot}},\
  }\bibfield  {title} {\bibinfo {title} {Towards scalable bosonic quantum error
  correction},\ }\href {https://doi.org/10.1088/2058-9565/ab98a5} {\bibfield
  {journal} {\bibinfo  {journal} {Quantum Science and Technology}\ }\textbf
  {\bibinfo {volume} {5}},\ \bibinfo {pages} {043001} (\bibinfo {year}
  {2020})}\BibitemShut {NoStop}%
\bibitem [{\citenamefont
  {Noh}(2021)}]{noh2021quantumcomputationcommunicationbosonic}%
  \BibitemOpen
  \bibfield  {author} {\bibinfo {author} {\bibfnamefont {K.}~\bibnamefont
  {Noh}},\ }\href {https://arxiv.org/abs/2103.09445} {\bibinfo {title} {Quantum
  computation and communication in bosonic systems}} (\bibinfo {year} {2021}),\
  \Eprint {https://arxiv.org/abs/2103.09445} {arXiv:2103.09445 [quant-ph]}
  \BibitemShut {NoStop}%
\bibitem [{\citenamefont {Cai}\ \emph {et~al.}(2021)\citenamefont {Cai},
  \citenamefont {Ma}, \citenamefont {Wang}, \citenamefont {Zou},\ and\
  \citenamefont {Sun}}]{Cai_2021}%
  \BibitemOpen
  \bibfield  {author} {\bibinfo {author} {\bibfnamefont {W.}~\bibnamefont
  {Cai}}, \bibinfo {author} {\bibfnamefont {Y.}~\bibnamefont {Ma}}, \bibinfo
  {author} {\bibfnamefont {W.}~\bibnamefont {Wang}}, \bibinfo {author}
  {\bibfnamefont {C.-L.}\ \bibnamefont {Zou}},\ and\ \bibinfo {author}
  {\bibfnamefont {L.}~\bibnamefont {Sun}},\ }\bibfield  {title} {\bibinfo
  {title} {Bosonic quantum error correction codes in superconducting quantum
  circuits},\ }\href {https://doi.org/10.1016/j.fmre.2020.12.006} {\bibfield
  {journal} {\bibinfo  {journal} {Fundamental Research}\ }\textbf {\bibinfo
  {volume} {1}},\ \bibinfo {pages} {50–67} (\bibinfo {year}
  {2021})}\BibitemShut {NoStop}%
\bibitem [{\citenamefont {Joshi}\ \emph {et~al.}(2021)\citenamefont {Joshi},
  \citenamefont {Noh},\ and\ \citenamefont {Gao}}]{joshi2021quantum}%
  \BibitemOpen
  \bibfield  {author} {\bibinfo {author} {\bibfnamefont {A.}~\bibnamefont
  {Joshi}}, \bibinfo {author} {\bibfnamefont {K.}~\bibnamefont {Noh}},\ and\
  \bibinfo {author} {\bibfnamefont {Y.~Y.}\ \bibnamefont {Gao}},\ }\bibfield
  {title} {\bibinfo {title} {Quantum information processing with bosonic qubits
  in circuit qed},\ }\href@noop {} {\bibfield  {journal} {\bibinfo  {journal}
  {Quantum Science and Technology}\ }\textbf {\bibinfo {volume} {6}},\ \bibinfo
  {pages} {033001} (\bibinfo {year} {2021})}\BibitemShut {NoStop}%
\bibitem [{\citenamefont {Cochrane}\ \emph {et~al.}(1999)\citenamefont
  {Cochrane}, \citenamefont {Milburn},\ and\ \citenamefont
  {Munro}}]{cochrane1999macroscopically}%
  \BibitemOpen
  \bibfield  {author} {\bibinfo {author} {\bibfnamefont {P.~T.}\ \bibnamefont
  {Cochrane}}, \bibinfo {author} {\bibfnamefont {G.~J.}\ \bibnamefont
  {Milburn}},\ and\ \bibinfo {author} {\bibfnamefont {W.~J.}\ \bibnamefont
  {Munro}},\ }\bibfield  {title} {\bibinfo {title} {Macroscopically distinct
  quantum-superposition states as a bosonic code for amplitude damping},\
  }\href {https://doi.org/10.1103/PhysRevA.59.2631} {\bibfield  {journal}
  {\bibinfo  {journal} {Phys. Rev. A}\ }\textbf {\bibinfo {volume} {59}},\
  \bibinfo {pages} {2631} (\bibinfo {year} {1999})}\BibitemShut {NoStop}%
\bibitem [{\citenamefont {Mirrahimi}\ \emph
  {et~al.}(2014{\natexlab{a}})\citenamefont {Mirrahimi}, \citenamefont
  {Leghtas}, \citenamefont {Albert}, \citenamefont {Touzard}, \citenamefont
  {Schoelkopf}, \citenamefont {Jiang},\ and\ \citenamefont
  {Devoret}}]{Mirrahimi_2014}%
  \BibitemOpen
  \bibfield  {author} {\bibinfo {author} {\bibfnamefont {M.}~\bibnamefont
  {Mirrahimi}}, \bibinfo {author} {\bibfnamefont {Z.}~\bibnamefont {Leghtas}},
  \bibinfo {author} {\bibfnamefont {V.~V.}\ \bibnamefont {Albert}}, \bibinfo
  {author} {\bibfnamefont {S.}~\bibnamefont {Touzard}}, \bibinfo {author}
  {\bibfnamefont {R.~J.}\ \bibnamefont {Schoelkopf}}, \bibinfo {author}
  {\bibfnamefont {L.}~\bibnamefont {Jiang}},\ and\ \bibinfo {author}
  {\bibfnamefont {M.~H.}\ \bibnamefont {Devoret}},\ }\bibfield  {title}
  {\bibinfo {title} {Dynamically protected cat-qubits: a new paradigm for
  universal quantum computation},\ }\href
  {https://doi.org/10.1088/1367-2630/16/4/045014} {\bibfield  {journal}
  {\bibinfo  {journal} {New Journal of Physics}\ }\textbf {\bibinfo {volume}
  {16}},\ \bibinfo {pages} {045014} (\bibinfo {year}
  {2014}{\natexlab{a}})}\BibitemShut {NoStop}%
\bibitem [{\citenamefont {Leghtas}\ \emph {et~al.}(2015)\citenamefont
  {Leghtas}, \citenamefont {Touzard}, \citenamefont {Pop}, \citenamefont {Kou},
  \citenamefont {Vlastakis}, \citenamefont {Petrenko}, \citenamefont {Sliwa},
  \citenamefont {Narla}, \citenamefont {Shankar}, \citenamefont {Hatridge},
  \citenamefont {Reagor}, \citenamefont {Frunzio}, \citenamefont {Schoelkopf},
  \citenamefont {Mirrahimi},\ and\ \citenamefont
  {Devoret}}]{leghtas2015confining}%
  \BibitemOpen
  \bibfield  {author} {\bibinfo {author} {\bibfnamefont {Z.}~\bibnamefont
  {Leghtas}}, \bibinfo {author} {\bibfnamefont {S.}~\bibnamefont {Touzard}},
  \bibinfo {author} {\bibfnamefont {I.~M.}\ \bibnamefont {Pop}}, \bibinfo
  {author} {\bibfnamefont {A.}~\bibnamefont {Kou}}, \bibinfo {author}
  {\bibfnamefont {B.}~\bibnamefont {Vlastakis}}, \bibinfo {author}
  {\bibfnamefont {A.}~\bibnamefont {Petrenko}}, \bibinfo {author}
  {\bibfnamefont {K.~M.}\ \bibnamefont {Sliwa}}, \bibinfo {author}
  {\bibfnamefont {A.}~\bibnamefont {Narla}}, \bibinfo {author} {\bibfnamefont
  {S.}~\bibnamefont {Shankar}}, \bibinfo {author} {\bibfnamefont {M.~J.}\
  \bibnamefont {Hatridge}}, \bibinfo {author} {\bibfnamefont {M.}~\bibnamefont
  {Reagor}}, \bibinfo {author} {\bibfnamefont {L.}~\bibnamefont {Frunzio}},
  \bibinfo {author} {\bibfnamefont {R.~J.}\ \bibnamefont {Schoelkopf}},
  \bibinfo {author} {\bibfnamefont {M.}~\bibnamefont {Mirrahimi}},\ and\
  \bibinfo {author} {\bibfnamefont {M.~H.}\ \bibnamefont {Devoret}},\
  }\bibfield  {title} {\bibinfo {title} {Confining the state of light to a
  quantum manifold by engineered two-photon loss},\ }\href
  {https://doi.org/10.1126/science.aaa2085} {\bibfield  {journal} {\bibinfo
  {journal} {Science}\ }\textbf {\bibinfo {volume} {347}},\ \bibinfo {pages}
  {853} (\bibinfo {year} {2015})},\ \Eprint
  {https://arxiv.org/abs/https://www.science.org/doi/pdf/10.1126/science.aaa2085}
  {https://www.science.org/doi/pdf/10.1126/science.aaa2085} \BibitemShut
  {NoStop}%
\bibitem [{\citenamefont {Niu}\ \emph {et~al.}(2018)\citenamefont {Niu},
  \citenamefont {Chuang},\ and\ \citenamefont {Shapiro}}]{Niu_2018}%
  \BibitemOpen
  \bibfield  {author} {\bibinfo {author} {\bibfnamefont {M.~Y.}\ \bibnamefont
  {Niu}}, \bibinfo {author} {\bibfnamefont {I.~L.}\ \bibnamefont {Chuang}},\
  and\ \bibinfo {author} {\bibfnamefont {J.~H.}\ \bibnamefont {Shapiro}},\
  }\bibfield  {title} {\bibinfo {title} {Hardware-efficient bosonic quantum
  error-correcting codes based on symmetry operators},\ }\bibfield  {journal}
  {\bibinfo  {journal} {Physical Review A}\ }\textbf {\bibinfo {volume} {97}},\
  \href {https://doi.org/10.1103/physreva.97.032323}
  {10.1103/physreva.97.032323} (\bibinfo {year} {2018})\BibitemShut {NoStop}%
\bibitem [{\citenamefont {Lemonde}\ \emph {et~al.}(2024)\citenamefont
  {Lemonde}, \citenamefont {Lachance-Quirion}, \citenamefont {Duclos-Cianci},
  \citenamefont {Frattini}, \citenamefont {Hopfmueller}, \citenamefont
  {Gauvin-Ndiaye}, \citenamefont {Camirand-Lemyre},\ and\ \citenamefont
  {St-Jean}}]{lemonde2024hardwareefficientfaulttolerantquantum}%
  \BibitemOpen
  \bibfield  {author} {\bibinfo {author} {\bibfnamefont {M.-A.}\ \bibnamefont
  {Lemonde}}, \bibinfo {author} {\bibfnamefont {D.}~\bibnamefont
  {Lachance-Quirion}}, \bibinfo {author} {\bibfnamefont {G.}~\bibnamefont
  {Duclos-Cianci}}, \bibinfo {author} {\bibfnamefont {N.~E.}\ \bibnamefont
  {Frattini}}, \bibinfo {author} {\bibfnamefont {F.}~\bibnamefont
  {Hopfmueller}}, \bibinfo {author} {\bibfnamefont {C.}~\bibnamefont
  {Gauvin-Ndiaye}}, \bibinfo {author} {\bibfnamefont {J.}~\bibnamefont
  {Camirand-Lemyre}},\ and\ \bibinfo {author} {\bibfnamefont {P.}~\bibnamefont
  {St-Jean}},\ }\href {https://arxiv.org/abs/2409.05813} {\bibinfo {title}
  {Hardware-efficient fault tolerant quantum computing with bosonic grid states
  in superconducting circuits}} (\bibinfo {year} {2024}),\ \Eprint
  {https://arxiv.org/abs/2409.05813} {arXiv:2409.05813 [quant-ph]} \BibitemShut
  {NoStop}%
\bibitem [{\citenamefont {Mori}\ \emph {et~al.}(2024)\citenamefont {Mori},
  \citenamefont {Matsuzaki}, \citenamefont {Endo},\ and\ \citenamefont
  {Kawabata}}]{mori2024hardwareefficientbosonicquantumcomputing}%
  \BibitemOpen
  \bibfield  {author} {\bibinfo {author} {\bibfnamefont {Y.}~\bibnamefont
  {Mori}}, \bibinfo {author} {\bibfnamefont {Y.}~\bibnamefont {Matsuzaki}},
  \bibinfo {author} {\bibfnamefont {S.}~\bibnamefont {Endo}},\ and\ \bibinfo
  {author} {\bibfnamefont {S.}~\bibnamefont {Kawabata}},\ }\href
  {https://arxiv.org/abs/2403.00291} {\bibinfo {title} {Hardware-efficient
  bosonic quantum computing with photon-loss detection capability}} (\bibinfo
  {year} {2024}),\ \Eprint {https://arxiv.org/abs/2403.00291} {arXiv:2403.00291
  [quant-ph]} \BibitemShut {NoStop}%
\bibitem [{\citenamefont {Putterman}\ \emph {et~al.}(2025)\citenamefont
  {Putterman}, \citenamefont {Noh}, \citenamefont {Hann}, \citenamefont
  {MacCabe}, \citenamefont {Aghaeimeibodi}, \citenamefont {Patel},
  \citenamefont {Lee}, \citenamefont {Jones}, \citenamefont {Moradinejad},
  \citenamefont {Rodriguez}, \citenamefont {Mahuli}, \citenamefont {Rose},
  \citenamefont {Owens}, \citenamefont {Levine}, \citenamefont {Rosenfeld},
  \citenamefont {Reinhold}, \citenamefont {Moncelsi}, \citenamefont {Alcid},
  \citenamefont {Alidoust}, \citenamefont {Arrangoiz-Arriola}, \citenamefont
  {Barnett}, \citenamefont {Bienias}, \citenamefont {Carson}, \citenamefont
  {Chen}, \citenamefont {Chen}, \citenamefont {Chinkezian}, \citenamefont
  {Chisholm}, \citenamefont {Chou}, \citenamefont {Clerk}, \citenamefont
  {Clifford}, \citenamefont {Cosmic}, \citenamefont {Curiel}, \citenamefont
  {Davis}, \citenamefont {DeLorenzo}, \citenamefont {D’Ewart}, \citenamefont
  {Diky}, \citenamefont {D’Souza}, \citenamefont {Dumitrescu}, \citenamefont
  {Eisenmann}, \citenamefont {Elkhouly}, \citenamefont {Evenbly}, \citenamefont
  {Fang}, \citenamefont {Fang}, \citenamefont {Fling}, \citenamefont {Fon},
  \citenamefont {Garcia}, \citenamefont {Gorshkov}, \citenamefont {Grant},
  \citenamefont {Gray}, \citenamefont {Grimberg}, \citenamefont {Grimsmo},
  \citenamefont {Haim}, \citenamefont {Hand}, \citenamefont {He}, \citenamefont
  {Hernandez}, \citenamefont {Hover}, \citenamefont {Hung}, \citenamefont
  {Hunt}, \citenamefont {Iverson}, \citenamefont {Jarrige}, \citenamefont
  {Jaskula}, \citenamefont {Jiang}, \citenamefont {Kalaee}, \citenamefont
  {Karabalin}, \citenamefont {Karalekas}, \citenamefont {Keller}, \citenamefont
  {Khalajhedayati}, \citenamefont {Kubica}, \citenamefont {Lee}, \citenamefont
  {Leroux}, \citenamefont {Lieu}, \citenamefont {Ly}, \citenamefont {Madrigal},
  \citenamefont {Marcaud}, \citenamefont {McCabe}, \citenamefont {Miles},
  \citenamefont {Milsted}, \citenamefont {Minguzzi}, \citenamefont {Mishra},
  \citenamefont {Mukherjee}, \citenamefont {Naghiloo}, \citenamefont
  {Oblepias}, \citenamefont {Ortuno}, \citenamefont {Pagdilao}, \citenamefont
  {Pancotti}, \citenamefont {Panduro}, \citenamefont {Paquette}, \citenamefont
  {Park}, \citenamefont {Peairs}, \citenamefont {Perello}, \citenamefont
  {Peterson}, \citenamefont {Ponte}, \citenamefont {Preskill}, \citenamefont
  {Qiao}, \citenamefont {Refael}, \citenamefont {Resnick}, \citenamefont
  {Retzker}, \citenamefont {Reyna}, \citenamefont {Runyan}, \citenamefont
  {Ryan}, \citenamefont {Sahmoud}, \citenamefont {Sanchez}, \citenamefont
  {Sanil}, \citenamefont {Sankar}, \citenamefont {Sato}, \citenamefont
  {Scaffidi}, \citenamefont {Siavoshi}, \citenamefont {Sivarajah},
  \citenamefont {Skogland}, \citenamefont {Su}, \citenamefont {Swenson},
  \citenamefont {Teo}, \citenamefont {Tomada}, \citenamefont {Torlai},
  \citenamefont {Wollack}, \citenamefont {Ye}, \citenamefont {Zerrudo},
  \citenamefont {Zhang}, \citenamefont {Brandão}, \citenamefont {Matheny},\
  and\ \citenamefont {Painter}}]{Putterman_2025}%
  \BibitemOpen
  \bibfield  {author} {\bibinfo {author} {\bibfnamefont {H.}~\bibnamefont
  {Putterman}}, \bibinfo {author} {\bibfnamefont {K.}~\bibnamefont {Noh}},
  \bibinfo {author} {\bibfnamefont {C.~T.}\ \bibnamefont {Hann}}, \bibinfo
  {author} {\bibfnamefont {G.~S.}\ \bibnamefont {MacCabe}}, \bibinfo {author}
  {\bibfnamefont {S.}~\bibnamefont {Aghaeimeibodi}}, \bibinfo {author}
  {\bibfnamefont {R.~N.}\ \bibnamefont {Patel}}, \bibinfo {author}
  {\bibfnamefont {M.}~\bibnamefont {Lee}}, \bibinfo {author} {\bibfnamefont
  {W.~M.}\ \bibnamefont {Jones}}, \bibinfo {author} {\bibfnamefont
  {H.}~\bibnamefont {Moradinejad}}, \bibinfo {author} {\bibfnamefont
  {R.}~\bibnamefont {Rodriguez}}, \bibinfo {author} {\bibfnamefont
  {N.}~\bibnamefont {Mahuli}}, \bibinfo {author} {\bibfnamefont
  {J.}~\bibnamefont {Rose}}, \bibinfo {author} {\bibfnamefont {J.~C.}\
  \bibnamefont {Owens}}, \bibinfo {author} {\bibfnamefont {H.}~\bibnamefont
  {Levine}}, \bibinfo {author} {\bibfnamefont {E.}~\bibnamefont {Rosenfeld}},
  \bibinfo {author} {\bibfnamefont {P.}~\bibnamefont {Reinhold}}, \bibinfo
  {author} {\bibfnamefont {L.}~\bibnamefont {Moncelsi}}, \bibinfo {author}
  {\bibfnamefont {J.~A.}\ \bibnamefont {Alcid}}, \bibinfo {author}
  {\bibfnamefont {N.}~\bibnamefont {Alidoust}}, \bibinfo {author}
  {\bibfnamefont {P.}~\bibnamefont {Arrangoiz-Arriola}}, \bibinfo {author}
  {\bibfnamefont {J.}~\bibnamefont {Barnett}}, \bibinfo {author} {\bibfnamefont
  {P.}~\bibnamefont {Bienias}}, \bibinfo {author} {\bibfnamefont {H.~A.}\
  \bibnamefont {Carson}}, \bibinfo {author} {\bibfnamefont {C.}~\bibnamefont
  {Chen}}, \bibinfo {author} {\bibfnamefont {L.}~\bibnamefont {Chen}}, \bibinfo
  {author} {\bibfnamefont {H.}~\bibnamefont {Chinkezian}}, \bibinfo {author}
  {\bibfnamefont {E.~M.}\ \bibnamefont {Chisholm}}, \bibinfo {author}
  {\bibfnamefont {M.-H.}\ \bibnamefont {Chou}}, \bibinfo {author}
  {\bibfnamefont {A.}~\bibnamefont {Clerk}}, \bibinfo {author} {\bibfnamefont
  {A.}~\bibnamefont {Clifford}}, \bibinfo {author} {\bibfnamefont
  {R.}~\bibnamefont {Cosmic}}, \bibinfo {author} {\bibfnamefont {A.~V.}\
  \bibnamefont {Curiel}}, \bibinfo {author} {\bibfnamefont {E.}~\bibnamefont
  {Davis}}, \bibinfo {author} {\bibfnamefont {L.}~\bibnamefont {DeLorenzo}},
  \bibinfo {author} {\bibfnamefont {J.~M.}\ \bibnamefont {D’Ewart}}, \bibinfo
  {author} {\bibfnamefont {A.}~\bibnamefont {Diky}}, \bibinfo {author}
  {\bibfnamefont {N.}~\bibnamefont {D’Souza}}, \bibinfo {author}
  {\bibfnamefont {P.~T.}\ \bibnamefont {Dumitrescu}}, \bibinfo {author}
  {\bibfnamefont {S.}~\bibnamefont {Eisenmann}}, \bibinfo {author}
  {\bibfnamefont {E.}~\bibnamefont {Elkhouly}}, \bibinfo {author}
  {\bibfnamefont {G.}~\bibnamefont {Evenbly}}, \bibinfo {author} {\bibfnamefont
  {M.~T.}\ \bibnamefont {Fang}}, \bibinfo {author} {\bibfnamefont
  {Y.}~\bibnamefont {Fang}}, \bibinfo {author} {\bibfnamefont {M.~J.}\
  \bibnamefont {Fling}}, \bibinfo {author} {\bibfnamefont {W.}~\bibnamefont
  {Fon}}, \bibinfo {author} {\bibfnamefont {G.}~\bibnamefont {Garcia}},
  \bibinfo {author} {\bibfnamefont {A.~V.}\ \bibnamefont {Gorshkov}}, \bibinfo
  {author} {\bibfnamefont {J.~A.}\ \bibnamefont {Grant}}, \bibinfo {author}
  {\bibfnamefont {M.~J.}\ \bibnamefont {Gray}}, \bibinfo {author}
  {\bibfnamefont {S.}~\bibnamefont {Grimberg}}, \bibinfo {author}
  {\bibfnamefont {A.~L.}\ \bibnamefont {Grimsmo}}, \bibinfo {author}
  {\bibfnamefont {A.}~\bibnamefont {Haim}}, \bibinfo {author} {\bibfnamefont
  {J.}~\bibnamefont {Hand}}, \bibinfo {author} {\bibfnamefont {Y.}~\bibnamefont
  {He}}, \bibinfo {author} {\bibfnamefont {M.}~\bibnamefont {Hernandez}},
  \bibinfo {author} {\bibfnamefont {D.}~\bibnamefont {Hover}}, \bibinfo
  {author} {\bibfnamefont {J.~S.~C.}\ \bibnamefont {Hung}}, \bibinfo {author}
  {\bibfnamefont {M.}~\bibnamefont {Hunt}}, \bibinfo {author} {\bibfnamefont
  {J.}~\bibnamefont {Iverson}}, \bibinfo {author} {\bibfnamefont
  {I.}~\bibnamefont {Jarrige}}, \bibinfo {author} {\bibfnamefont {J.-C.}\
  \bibnamefont {Jaskula}}, \bibinfo {author} {\bibfnamefont {L.}~\bibnamefont
  {Jiang}}, \bibinfo {author} {\bibfnamefont {M.}~\bibnamefont {Kalaee}},
  \bibinfo {author} {\bibfnamefont {R.}~\bibnamefont {Karabalin}}, \bibinfo
  {author} {\bibfnamefont {P.~J.}\ \bibnamefont {Karalekas}}, \bibinfo {author}
  {\bibfnamefont {A.~J.}\ \bibnamefont {Keller}}, \bibinfo {author}
  {\bibfnamefont {A.}~\bibnamefont {Khalajhedayati}}, \bibinfo {author}
  {\bibfnamefont {A.}~\bibnamefont {Kubica}}, \bibinfo {author} {\bibfnamefont
  {H.}~\bibnamefont {Lee}}, \bibinfo {author} {\bibfnamefont {C.}~\bibnamefont
  {Leroux}}, \bibinfo {author} {\bibfnamefont {S.}~\bibnamefont {Lieu}},
  \bibinfo {author} {\bibfnamefont {V.}~\bibnamefont {Ly}}, \bibinfo {author}
  {\bibfnamefont {K.~V.}\ \bibnamefont {Madrigal}}, \bibinfo {author}
  {\bibfnamefont {G.}~\bibnamefont {Marcaud}}, \bibinfo {author} {\bibfnamefont
  {G.}~\bibnamefont {McCabe}}, \bibinfo {author} {\bibfnamefont
  {C.}~\bibnamefont {Miles}}, \bibinfo {author} {\bibfnamefont
  {A.}~\bibnamefont {Milsted}}, \bibinfo {author} {\bibfnamefont
  {J.}~\bibnamefont {Minguzzi}}, \bibinfo {author} {\bibfnamefont
  {A.}~\bibnamefont {Mishra}}, \bibinfo {author} {\bibfnamefont
  {B.}~\bibnamefont {Mukherjee}}, \bibinfo {author} {\bibfnamefont
  {M.}~\bibnamefont {Naghiloo}}, \bibinfo {author} {\bibfnamefont
  {E.}~\bibnamefont {Oblepias}}, \bibinfo {author} {\bibfnamefont
  {G.}~\bibnamefont {Ortuno}}, \bibinfo {author} {\bibfnamefont
  {J.}~\bibnamefont {Pagdilao}}, \bibinfo {author} {\bibfnamefont
  {N.}~\bibnamefont {Pancotti}}, \bibinfo {author} {\bibfnamefont
  {A.}~\bibnamefont {Panduro}}, \bibinfo {author} {\bibfnamefont
  {J.}~\bibnamefont {Paquette}}, \bibinfo {author} {\bibfnamefont
  {M.}~\bibnamefont {Park}}, \bibinfo {author} {\bibfnamefont {G.~A.}\
  \bibnamefont {Peairs}}, \bibinfo {author} {\bibfnamefont {D.}~\bibnamefont
  {Perello}}, \bibinfo {author} {\bibfnamefont {E.~C.}\ \bibnamefont
  {Peterson}}, \bibinfo {author} {\bibfnamefont {S.}~\bibnamefont {Ponte}},
  \bibinfo {author} {\bibfnamefont {J.}~\bibnamefont {Preskill}}, \bibinfo
  {author} {\bibfnamefont {J.}~\bibnamefont {Qiao}}, \bibinfo {author}
  {\bibfnamefont {G.}~\bibnamefont {Refael}}, \bibinfo {author} {\bibfnamefont
  {R.}~\bibnamefont {Resnick}}, \bibinfo {author} {\bibfnamefont
  {A.}~\bibnamefont {Retzker}}, \bibinfo {author} {\bibfnamefont {O.~A.}\
  \bibnamefont {Reyna}}, \bibinfo {author} {\bibfnamefont {M.}~\bibnamefont
  {Runyan}}, \bibinfo {author} {\bibfnamefont {C.~A.}\ \bibnamefont {Ryan}},
  \bibinfo {author} {\bibfnamefont {A.}~\bibnamefont {Sahmoud}}, \bibinfo
  {author} {\bibfnamefont {E.}~\bibnamefont {Sanchez}}, \bibinfo {author}
  {\bibfnamefont {R.}~\bibnamefont {Sanil}}, \bibinfo {author} {\bibfnamefont
  {K.}~\bibnamefont {Sankar}}, \bibinfo {author} {\bibfnamefont
  {Y.}~\bibnamefont {Sato}}, \bibinfo {author} {\bibfnamefont {T.}~\bibnamefont
  {Scaffidi}}, \bibinfo {author} {\bibfnamefont {S.}~\bibnamefont {Siavoshi}},
  \bibinfo {author} {\bibfnamefont {P.}~\bibnamefont {Sivarajah}}, \bibinfo
  {author} {\bibfnamefont {T.}~\bibnamefont {Skogland}}, \bibinfo {author}
  {\bibfnamefont {C.-J.}\ \bibnamefont {Su}}, \bibinfo {author} {\bibfnamefont
  {L.~J.}\ \bibnamefont {Swenson}}, \bibinfo {author} {\bibfnamefont {S.~M.}\
  \bibnamefont {Teo}}, \bibinfo {author} {\bibfnamefont {A.}~\bibnamefont
  {Tomada}}, \bibinfo {author} {\bibfnamefont {G.}~\bibnamefont {Torlai}},
  \bibinfo {author} {\bibfnamefont {E.~A.}\ \bibnamefont {Wollack}}, \bibinfo
  {author} {\bibfnamefont {Y.}~\bibnamefont {Ye}}, \bibinfo {author}
  {\bibfnamefont {J.~A.}\ \bibnamefont {Zerrudo}}, \bibinfo {author}
  {\bibfnamefont {K.}~\bibnamefont {Zhang}}, \bibinfo {author} {\bibfnamefont
  {F.~G. S.~L.}\ \bibnamefont {Brandão}}, \bibinfo {author} {\bibfnamefont
  {M.~H.}\ \bibnamefont {Matheny}},\ and\ \bibinfo {author} {\bibfnamefont
  {O.}~\bibnamefont {Painter}},\ }\bibfield  {title} {\bibinfo {title}
  {Hardware-efficient quantum error correction via concatenated bosonic
  qubits},\ }\href {https://doi.org/10.1038/s41586-025-08642-7} {\bibfield
  {journal} {\bibinfo  {journal} {Nature}\ }\textbf {\bibinfo {volume} {638}},\
  \bibinfo {pages} {927–934} (\bibinfo {year} {2025})}\BibitemShut {NoStop}%
\bibitem [{\citenamefont {Gottesman}\ \emph {et~al.}(2001)\citenamefont
  {Gottesman}, \citenamefont {Kitaev},\ and\ \citenamefont
  {Preskill}}]{Gottesman_2001}%
  \BibitemOpen
  \bibfield  {author} {\bibinfo {author} {\bibfnamefont {D.}~\bibnamefont
  {Gottesman}}, \bibinfo {author} {\bibfnamefont {A.}~\bibnamefont {Kitaev}},\
  and\ \bibinfo {author} {\bibfnamefont {J.}~\bibnamefont {Preskill}},\
  }\bibfield  {title} {\bibinfo {title} {Encoding a qubit in an oscillator},\
  }\bibfield  {journal} {\bibinfo  {journal} {Physical Review A}\ }\textbf
  {\bibinfo {volume} {64}},\ \href {https://doi.org/10.1103/physreva.64.012310}
  {10.1103/physreva.64.012310} (\bibinfo {year} {2001})\BibitemShut {NoStop}%
\bibitem [{\citenamefont {Conrad}(2024)}]{Conrad_GKP}%
  \BibitemOpen
  \bibfield  {author} {\bibinfo {author} {\bibfnamefont {J.}~\bibnamefont
  {Conrad}},\ }\emph {\bibinfo {title} {The fabulous world of GKP codes}},\
  \href {https://doi.org/10.17169/REFUBIUM-45505} {Ph.D. thesis} (\bibinfo
  {year} {2024})\BibitemShut {NoStop}%
\bibitem [{\citenamefont {Brady}\ \emph {et~al.}(2024)\citenamefont {Brady},
  \citenamefont {Eickbusch}, \citenamefont {Singh}, \citenamefont {Wu},\ and\
  \citenamefont {Zhuang}}]{Brady_2024}%
  \BibitemOpen
  \bibfield  {author} {\bibinfo {author} {\bibfnamefont {A.~J.}\ \bibnamefont
  {Brady}}, \bibinfo {author} {\bibfnamefont {A.}~\bibnamefont {Eickbusch}},
  \bibinfo {author} {\bibfnamefont {S.}~\bibnamefont {Singh}}, \bibinfo
  {author} {\bibfnamefont {J.}~\bibnamefont {Wu}},\ and\ \bibinfo {author}
  {\bibfnamefont {Q.}~\bibnamefont {Zhuang}},\ }\bibfield  {title} {\bibinfo
  {title} {Advances in bosonic quantum error correction with
  gottesman–kitaev–preskill codes: Theory, engineering and applications},\
  }\href {https://doi.org/10.1016/j.pquantelec.2023.100496} {\bibfield
  {journal} {\bibinfo  {journal} {Progress in Quantum Electronics}\ }\textbf
  {\bibinfo {volume} {93}},\ \bibinfo {pages} {100496} (\bibinfo {year}
  {2024})}\BibitemShut {NoStop}%
\bibitem [{\citenamefont {Flühmann}\ \emph {et~al.}(2019)\citenamefont
  {Flühmann}, \citenamefont {Nguyen}, \citenamefont {Marinelli}, \citenamefont
  {Negnevitsky}, \citenamefont {Mehta},\ and\ \citenamefont
  {Home}}]{fluhmann_encoding_2019}%
  \BibitemOpen
  \bibfield  {author} {\bibinfo {author} {\bibfnamefont {C.}~\bibnamefont
  {Flühmann}}, \bibinfo {author} {\bibfnamefont {T.~L.}\ \bibnamefont
  {Nguyen}}, \bibinfo {author} {\bibfnamefont {M.}~\bibnamefont {Marinelli}},
  \bibinfo {author} {\bibfnamefont {V.}~\bibnamefont {Negnevitsky}}, \bibinfo
  {author} {\bibfnamefont {K.}~\bibnamefont {Mehta}},\ and\ \bibinfo {author}
  {\bibfnamefont {J.~P.}\ \bibnamefont {Home}},\ }\bibfield  {title} {\bibinfo
  {title} {Encoding a qubit in a trapped-ion mechanical oscillator},\ }\href
  {https://doi.org/10.1038/s41586-019-0960-6} {\bibfield  {journal} {\bibinfo
  {journal} {Nature}\ }\textbf {\bibinfo {volume} {566}},\ \bibinfo {pages}
  {513} (\bibinfo {year} {2019})}\BibitemShut {NoStop}%
\bibitem [{\citenamefont {Campagne-Ibarcq}\ \emph {et~al.}(2020)\citenamefont
  {Campagne-Ibarcq}, \citenamefont {Eickbusch}, \citenamefont {Touzard},
  \citenamefont {Zalys-Geller}, \citenamefont {Frattini}, \citenamefont
  {Sivak}, \citenamefont {Reinhold}, \citenamefont {Puri}, \citenamefont
  {Shankar}, \citenamefont {Schoelkopf} \emph {et~al.}}]{campagne2020quantum}%
  \BibitemOpen
  \bibfield  {author} {\bibinfo {author} {\bibfnamefont {P.}~\bibnamefont
  {Campagne-Ibarcq}}, \bibinfo {author} {\bibfnamefont {A.}~\bibnamefont
  {Eickbusch}}, \bibinfo {author} {\bibfnamefont {S.}~\bibnamefont {Touzard}},
  \bibinfo {author} {\bibfnamefont {E.}~\bibnamefont {Zalys-Geller}}, \bibinfo
  {author} {\bibfnamefont {N.~E.}\ \bibnamefont {Frattini}}, \bibinfo {author}
  {\bibfnamefont {V.~V.}\ \bibnamefont {Sivak}}, \bibinfo {author}
  {\bibfnamefont {P.}~\bibnamefont {Reinhold}}, \bibinfo {author}
  {\bibfnamefont {S.}~\bibnamefont {Puri}}, \bibinfo {author} {\bibfnamefont
  {S.}~\bibnamefont {Shankar}}, \bibinfo {author} {\bibfnamefont {R.~J.}\
  \bibnamefont {Schoelkopf}}, \emph {et~al.},\ }\bibfield  {title} {\bibinfo
  {title} {Quantum error correction of a qubit encoded in grid states of an
  oscillator},\ }\href@noop {} {\bibfield  {journal} {\bibinfo  {journal}
  {Nature}\ }\textbf {\bibinfo {volume} {584}},\ \bibinfo {pages} {368}
  (\bibinfo {year} {2020})}\BibitemShut {NoStop}%
\bibitem [{\citenamefont {Michael}\ \emph {et~al.}(2016)\citenamefont
  {Michael}, \citenamefont {Silveri}, \citenamefont {Brierley}, \citenamefont
  {Albert}, \citenamefont {Salmilehto}, \citenamefont {Jiang},\ and\
  \citenamefont {Girvin}}]{michael2016new}%
  \BibitemOpen
  \bibfield  {author} {\bibinfo {author} {\bibfnamefont {M.~H.}\ \bibnamefont
  {Michael}}, \bibinfo {author} {\bibfnamefont {M.}~\bibnamefont {Silveri}},
  \bibinfo {author} {\bibfnamefont {R.~T.}\ \bibnamefont {Brierley}}, \bibinfo
  {author} {\bibfnamefont {V.~V.}\ \bibnamefont {Albert}}, \bibinfo {author}
  {\bibfnamefont {J.}~\bibnamefont {Salmilehto}}, \bibinfo {author}
  {\bibfnamefont {L.}~\bibnamefont {Jiang}},\ and\ \bibinfo {author}
  {\bibfnamefont {S.~M.}\ \bibnamefont {Girvin}},\ }\bibfield  {title}
  {\bibinfo {title} {New class of quantum error-correcting codes for a bosonic
  mode},\ }\href {https://doi.org/10.1103/PhysRevX.6.031006} {\bibfield
  {journal} {\bibinfo  {journal} {Phys. Rev. X}\ }\textbf {\bibinfo {volume}
  {6}},\ \bibinfo {pages} {031006} (\bibinfo {year} {2016})}\BibitemShut
  {NoStop}%
\bibitem [{\citenamefont {Hu}\ \emph {et~al.}(2019)\citenamefont {Hu},
  \citenamefont {Ma}, \citenamefont {Cai}, \citenamefont {Mu}, \citenamefont
  {Xu}, \citenamefont {Wang}, \citenamefont {Wu}, \citenamefont {Wang},
  \citenamefont {Song}, \citenamefont {Zou} \emph {et~al.}}]{hu2019quantum}%
  \BibitemOpen
  \bibfield  {author} {\bibinfo {author} {\bibfnamefont {L.}~\bibnamefont
  {Hu}}, \bibinfo {author} {\bibfnamefont {Y.}~\bibnamefont {Ma}}, \bibinfo
  {author} {\bibfnamefont {W.}~\bibnamefont {Cai}}, \bibinfo {author}
  {\bibfnamefont {X.}~\bibnamefont {Mu}}, \bibinfo {author} {\bibfnamefont
  {Y.}~\bibnamefont {Xu}}, \bibinfo {author} {\bibfnamefont {W.}~\bibnamefont
  {Wang}}, \bibinfo {author} {\bibfnamefont {Y.}~\bibnamefont {Wu}}, \bibinfo
  {author} {\bibfnamefont {H.}~\bibnamefont {Wang}}, \bibinfo {author}
  {\bibfnamefont {Y.}~\bibnamefont {Song}}, \bibinfo {author} {\bibfnamefont
  {C.-L.}\ \bibnamefont {Zou}}, \emph {et~al.},\ }\bibfield  {title} {\bibinfo
  {title} {Quantum error correction and universal gate set operation on a
  binomial bosonic logical qubit},\ }\href@noop {} {\bibfield  {journal}
  {\bibinfo  {journal} {Nature Physics}\ }\textbf {\bibinfo {volume} {15}},\
  \bibinfo {pages} {503} (\bibinfo {year} {2019})}\BibitemShut {NoStop}%
\bibitem [{\citenamefont {Leghtas}\ \emph {et~al.}(2013)\citenamefont
  {Leghtas}, \citenamefont {Kirchmair}, \citenamefont {Vlastakis},
  \citenamefont {Schoelkopf}, \citenamefont {Devoret},\ and\ \citenamefont
  {Mirrahimi}}]{leghtas2013hardware}%
  \BibitemOpen
  \bibfield  {author} {\bibinfo {author} {\bibfnamefont {Z.}~\bibnamefont
  {Leghtas}}, \bibinfo {author} {\bibfnamefont {G.}~\bibnamefont {Kirchmair}},
  \bibinfo {author} {\bibfnamefont {B.}~\bibnamefont {Vlastakis}}, \bibinfo
  {author} {\bibfnamefont {R.~J.}\ \bibnamefont {Schoelkopf}}, \bibinfo
  {author} {\bibfnamefont {M.~H.}\ \bibnamefont {Devoret}},\ and\ \bibinfo
  {author} {\bibfnamefont {M.}~\bibnamefont {Mirrahimi}},\ }\bibfield  {title}
  {\bibinfo {title} {Hardware-efficient autonomous quantum memory protection},\
  }\href@noop {} {\bibfield  {journal} {\bibinfo  {journal} {Physical Review
  Letters}\ }\textbf {\bibinfo {volume} {111}},\ \bibinfo {pages} {120501}
  (\bibinfo {year} {2013})}\BibitemShut {NoStop}%
\bibitem [{\citenamefont {Mirrahimi}\ \emph
  {et~al.}(2014{\natexlab{b}})\citenamefont {Mirrahimi}, \citenamefont
  {Leghtas}, \citenamefont {Albert}, \citenamefont {Touzard}, \citenamefont
  {Schoelkopf}, \citenamefont {Jiang},\ and\ \citenamefont
  {Devoret}}]{mirrahimi2014dynamically}%
  \BibitemOpen
  \bibfield  {author} {\bibinfo {author} {\bibfnamefont {M.}~\bibnamefont
  {Mirrahimi}}, \bibinfo {author} {\bibfnamefont {Z.}~\bibnamefont {Leghtas}},
  \bibinfo {author} {\bibfnamefont {V.~V.}\ \bibnamefont {Albert}}, \bibinfo
  {author} {\bibfnamefont {S.}~\bibnamefont {Touzard}}, \bibinfo {author}
  {\bibfnamefont {R.~J.}\ \bibnamefont {Schoelkopf}}, \bibinfo {author}
  {\bibfnamefont {L.}~\bibnamefont {Jiang}},\ and\ \bibinfo {author}
  {\bibfnamefont {M.~H.}\ \bibnamefont {Devoret}},\ }\bibfield  {title}
  {\bibinfo {title} {Dynamically protected cat-qubits: a new paradigm for
  universal quantum computation},\ }\href@noop {} {\bibfield  {journal}
  {\bibinfo  {journal} {New Journal of Physics}\ }\textbf {\bibinfo {volume}
  {16}},\ \bibinfo {pages} {045014} (\bibinfo {year}
  {2014}{\natexlab{b}})}\BibitemShut {NoStop}%
\bibitem [{\citenamefont {Guillaud}\ \emph {et~al.}(2023)\citenamefont
  {Guillaud}, \citenamefont {Cohen},\ and\ \citenamefont
  {Mirrahimi}}]{guillaud2023quantum}%
  \BibitemOpen
  \bibfield  {author} {\bibinfo {author} {\bibfnamefont {J.}~\bibnamefont
  {Guillaud}}, \bibinfo {author} {\bibfnamefont {J.}~\bibnamefont {Cohen}},\
  and\ \bibinfo {author} {\bibfnamefont {M.}~\bibnamefont {Mirrahimi}},\
  }\bibfield  {title} {\bibinfo {title} {Quantum computation with cat qubits},\
  }\href@noop {} {\bibfield  {journal} {\bibinfo  {journal} {SciPost Physics
  Lecture Notes}\ ,\ \bibinfo {pages} {072}} (\bibinfo {year}
  {2023})}\BibitemShut {NoStop}%
\bibitem [{\citenamefont {Puri}\ \emph {et~al.}(2019)\citenamefont {Puri},
  \citenamefont {Grimm}, \citenamefont {Campagne-Ibarcq}, \citenamefont
  {Eickbusch}, \citenamefont {Noh}, \citenamefont {Roberts}, \citenamefont
  {Jiang}, \citenamefont {Mirrahimi}, \citenamefont {Devoret},\ and\
  \citenamefont {Girvin}}]{PhysRevX.9.041009}%
  \BibitemOpen
  \bibfield  {author} {\bibinfo {author} {\bibfnamefont {S.}~\bibnamefont
  {Puri}}, \bibinfo {author} {\bibfnamefont {A.}~\bibnamefont {Grimm}},
  \bibinfo {author} {\bibfnamefont {P.}~\bibnamefont {Campagne-Ibarcq}},
  \bibinfo {author} {\bibfnamefont {A.}~\bibnamefont {Eickbusch}}, \bibinfo
  {author} {\bibfnamefont {K.}~\bibnamefont {Noh}}, \bibinfo {author}
  {\bibfnamefont {G.}~\bibnamefont {Roberts}}, \bibinfo {author} {\bibfnamefont
  {L.}~\bibnamefont {Jiang}}, \bibinfo {author} {\bibfnamefont
  {M.}~\bibnamefont {Mirrahimi}}, \bibinfo {author} {\bibfnamefont {M.~H.}\
  \bibnamefont {Devoret}},\ and\ \bibinfo {author} {\bibfnamefont {S.~M.}\
  \bibnamefont {Girvin}},\ }\bibfield  {title} {\bibinfo {title} {Stabilized
  cat in a driven nonlinear cavity: A fault-tolerant error syndrome detector},\
  }\href {https://doi.org/10.1103/PhysRevX.9.041009} {\bibfield  {journal}
  {\bibinfo  {journal} {Phys. Rev. X}\ }\textbf {\bibinfo {volume} {9}},\
  \bibinfo {pages} {041009} (\bibinfo {year} {2019})}\BibitemShut {NoStop}%
\bibitem [{\citenamefont {Cools}(1992)}]{cools1992survey}%
  \BibitemOpen
  \bibfield  {author} {\bibinfo {author} {\bibfnamefont {R.}~\bibnamefont
  {Cools}},\ }\bibfield  {title} {\bibinfo {title} {A survey of methods for
  constructing cubature formulae},\ }in\ \href@noop {} {\emph {\bibinfo
  {booktitle} {Numerical Integration: Recent Developments, Software and
  Applications}}}\ (\bibinfo  {publisher} {Springer},\ \bibinfo {year} {1992})\
  pp.\ \bibinfo {pages} {1--24}\BibitemShut {NoStop}%
\bibitem [{\citenamefont {Cools}(2003)}]{cools2003encyclopaedia}%
  \BibitemOpen
  \bibfield  {author} {\bibinfo {author} {\bibfnamefont {R.}~\bibnamefont
  {Cools}},\ }\bibfield  {title} {\bibinfo {title} {An encyclopaedia of
  cubature formulas},\ }\href@noop {} {\bibfield  {journal} {\bibinfo
  {journal} {Journal of complexity}\ }\textbf {\bibinfo {volume} {19}},\
  \bibinfo {pages} {445} (\bibinfo {year} {2003})}\BibitemShut {NoStop}%
\bibitem [{\citenamefont {Stroud}(1971)}]{stroud1971approximate}%
  \BibitemOpen
  \bibfield  {author} {\bibinfo {author} {\bibfnamefont {A.~H.}\ \bibnamefont
  {Stroud}},\ }\href@noop {} {\emph {\bibinfo {title} {Approximate calculation
  of multiple integrals}}}\ (\bibinfo  {publisher} {Prentice Hall},\ \bibinfo
  {year} {1971})\BibitemShut {NoStop}%
\bibitem [{\citenamefont {Sawa}\ \emph {et~al.}(2019)\citenamefont {Sawa},
  \citenamefont {Hirao}, \citenamefont {Kageyama}, \citenamefont {Sawa},
  \citenamefont {Hirao},\ and\ \citenamefont {Kageyama}}]{sawa2019euclidean}%
  \BibitemOpen
  \bibfield  {author} {\bibinfo {author} {\bibfnamefont {M.}~\bibnamefont
  {Sawa}}, \bibinfo {author} {\bibfnamefont {M.}~\bibnamefont {Hirao}},
  \bibinfo {author} {\bibfnamefont {S.}~\bibnamefont {Kageyama}}, \bibinfo
  {author} {\bibfnamefont {M.}~\bibnamefont {Sawa}}, \bibinfo {author}
  {\bibfnamefont {M.}~\bibnamefont {Hirao}},\ and\ \bibinfo {author}
  {\bibfnamefont {S.}~\bibnamefont {Kageyama}},\ }\href@noop {} {\emph
  {\bibinfo {title} {Euclidean design theory}}}\ (\bibinfo  {publisher}
  {Springer},\ \bibinfo {year} {2019})\BibitemShut {NoStop}%
\bibitem [{\citenamefont {Bannai}\ \emph {et~al.}(2010)\citenamefont {Bannai},
  \citenamefont {Bannai}, \citenamefont {Hirao},\ and\ \citenamefont
  {Sawa}}]{bannai2010cubature}%
  \BibitemOpen
  \bibfield  {author} {\bibinfo {author} {\bibfnamefont {E.}~\bibnamefont
  {Bannai}}, \bibinfo {author} {\bibfnamefont {E.}~\bibnamefont {Bannai}},
  \bibinfo {author} {\bibfnamefont {M.}~\bibnamefont {Hirao}},\ and\ \bibinfo
  {author} {\bibfnamefont {M.}~\bibnamefont {Sawa}},\ }\bibfield  {title}
  {\bibinfo {title} {Cubature formulas in numerical analysis and euclidean
  tight designs},\ }\href@noop {} {\bibfield  {journal} {\bibinfo  {journal}
  {European Journal of Combinatorics}\ }\textbf {\bibinfo {volume} {31}},\
  \bibinfo {pages} {423} (\bibinfo {year} {2010})}\BibitemShut {NoStop}%
\bibitem [{\citenamefont {Jain}\ \emph {et~al.}(2024)\citenamefont {Jain},
  \citenamefont {Iosue}, \citenamefont {Barg},\ and\ \citenamefont
  {Albert}}]{jain2024quantum}%
  \BibitemOpen
  \bibfield  {author} {\bibinfo {author} {\bibfnamefont {S.~P.}\ \bibnamefont
  {Jain}}, \bibinfo {author} {\bibfnamefont {J.~T.}\ \bibnamefont {Iosue}},
  \bibinfo {author} {\bibfnamefont {A.}~\bibnamefont {Barg}},\ and\ \bibinfo
  {author} {\bibfnamefont {V.~V.}\ \bibnamefont {Albert}},\ }\bibfield  {title}
  {\bibinfo {title} {Quantum spherical codes},\ }\href@noop {} {\bibfield
  {journal} {\bibinfo  {journal} {Nature Physics}\ }\textbf {\bibinfo {volume}
  {20}},\ \bibinfo {pages} {1300} (\bibinfo {year} {2024})}\BibitemShut
  {NoStop}%
\bibitem [{\citenamefont {Knill}\ and\ \citenamefont
  {Laflamme}(1997)}]{knill1997theory}%
  \BibitemOpen
  \bibfield  {author} {\bibinfo {author} {\bibfnamefont {E.}~\bibnamefont
  {Knill}}\ and\ \bibinfo {author} {\bibfnamefont {R.}~\bibnamefont
  {Laflamme}},\ }\bibfield  {title} {\bibinfo {title} {Theory of quantum
  error-correcting codes},\ }\href@noop {} {\bibfield  {journal} {\bibinfo
  {journal} {Physical Review A}\ }\textbf {\bibinfo {volume} {55}},\ \bibinfo
  {pages} {900} (\bibinfo {year} {1997})}\BibitemShut {NoStop}%
\bibitem [{\citenamefont {Turchette}\ \emph {et~al.}(2000)\citenamefont
  {Turchette}, \citenamefont {Myatt}, \citenamefont {King}, \citenamefont
  {Sackett}, \citenamefont {Kielpinski}, \citenamefont {Itano}, \citenamefont
  {Monroe},\ and\ \citenamefont {Wineland}}]{turchette2000decoherence}%
  \BibitemOpen
  \bibfield  {author} {\bibinfo {author} {\bibfnamefont {Q.}~\bibnamefont
  {Turchette}}, \bibinfo {author} {\bibfnamefont {C.}~\bibnamefont {Myatt}},
  \bibinfo {author} {\bibfnamefont {B.}~\bibnamefont {King}}, \bibinfo {author}
  {\bibfnamefont {C.}~\bibnamefont {Sackett}}, \bibinfo {author} {\bibfnamefont
  {D.}~\bibnamefont {Kielpinski}}, \bibinfo {author} {\bibfnamefont
  {W.}~\bibnamefont {Itano}}, \bibinfo {author} {\bibfnamefont
  {C.}~\bibnamefont {Monroe}},\ and\ \bibinfo {author} {\bibfnamefont
  {D.}~\bibnamefont {Wineland}},\ }\bibfield  {title} {\bibinfo {title}
  {Decoherence and decay of motional quantum states of a trapped atom coupled
  to engineered reservoirs},\ }\href@noop {} {\bibfield  {journal} {\bibinfo
  {journal} {Physical Review A}\ }\textbf {\bibinfo {volume} {62}},\ \bibinfo
  {pages} {053807} (\bibinfo {year} {2000})}\BibitemShut {NoStop}%
\bibitem [{\citenamefont {Grimsmo}\ \emph {et~al.}(2020)\citenamefont
  {Grimsmo}, \citenamefont {Combes},\ and\ \citenamefont
  {Baragiola}}]{grimsmo2020quantum}%
  \BibitemOpen
  \bibfield  {author} {\bibinfo {author} {\bibfnamefont {A.~L.}\ \bibnamefont
  {Grimsmo}}, \bibinfo {author} {\bibfnamefont {J.}~\bibnamefont {Combes}},\
  and\ \bibinfo {author} {\bibfnamefont {B.~Q.}\ \bibnamefont {Baragiola}},\
  }\bibfield  {title} {\bibinfo {title} {Quantum computing with
  rotation-symmetric bosonic codes},\ }\href@noop {} {\bibfield  {journal}
  {\bibinfo  {journal} {Physical Review X}\ }\textbf {\bibinfo {volume} {10}},\
  \bibinfo {pages} {011058} (\bibinfo {year} {2020})}\BibitemShut {NoStop}%
\bibitem [{SM_()}]{SM_QCC}%
  \BibitemOpen
  \href@noop {} {}\bibinfo {note} {Supplemental Material contains a proof that
  QCCs satisfy the Knill-Laflamme conditions and a discussion of lower bounds
  on constellation size.}\BibitemShut {Stop}%
\bibitem [{\citenamefont {Delsarte}\ \emph {et~al.}(1977)\citenamefont
  {Delsarte}, \citenamefont {Goethals},\ and\ \citenamefont
  {Seidel}}]{delsarte_spherical_1977}%
  \BibitemOpen
  \bibfield  {author} {\bibinfo {author} {\bibfnamefont {P.}~\bibnamefont
  {Delsarte}}, \bibinfo {author} {\bibfnamefont {J.~M.}\ \bibnamefont
  {Goethals}},\ and\ \bibinfo {author} {\bibfnamefont {J.~J.}\ \bibnamefont
  {Seidel}},\ }\bibfield  {title} {\bibinfo {title} {Spherical codes and
  designs},\ }\href {https://doi.org/10.1007/BF03187604} {\bibfield  {journal}
  {\bibinfo  {journal} {Geometriae Dedicata}\ }\textbf {\bibinfo {volume}
  {6}},\ \bibinfo {pages} {363} (\bibinfo {year} {1977})}\BibitemShut {NoStop}%
\bibitem [{\citenamefont {M{\"o}ller}(1979)}]{moller1979lower}%
  \BibitemOpen
  \bibfield  {author} {\bibinfo {author} {\bibfnamefont {H.~M.}\ \bibnamefont
  {M{\"o}ller}},\ }\bibfield  {title} {\bibinfo {title} {Lower bounds for the
  number of nodes in cubature formulae},\ }\href@noop {} {\bibfield  {journal}
  {\bibinfo  {journal} {Numerische Integration: Tagung im Mathematischen
  Forschungsinstitut Oberwolfach vom 1. bis 7. Oktober 1978}\ ,\ \bibinfo
  {pages} {221}} (\bibinfo {year} {1979})}\BibitemShut {NoStop}%
\bibitem [{\citenamefont {Tchakaloff}(1957)}]{tchakaloff1957formules}%
  \BibitemOpen
  \bibfield  {author} {\bibinfo {author} {\bibfnamefont {V.}~\bibnamefont
  {Tchakaloff}},\ }\bibfield  {title} {\bibinfo {title} {Formules de cubatures
  m{\'e}caniques {\`a} coefficients non n{\'e}gatifs},\ }\href@noop {}
  {\bibfield  {journal} {\bibinfo  {journal} {Bull. Sci. Math}\ }\textbf
  {\bibinfo {volume} {81}},\ \bibinfo {pages} {123} (\bibinfo {year}
  {1957})}\BibitemShut {NoStop}%
\bibitem [{\citenamefont {Bannai}\ \emph {et~al.}(2015)\citenamefont {Bannai},
  \citenamefont {Bannai},\ and\ \citenamefont {Zhu}}]{bannai2015survey}%
  \BibitemOpen
  \bibfield  {author} {\bibinfo {author} {\bibfnamefont {E.}~\bibnamefont
  {Bannai}}, \bibinfo {author} {\bibfnamefont {E.}~\bibnamefont {Bannai}},\
  and\ \bibinfo {author} {\bibfnamefont {Y.}~\bibnamefont {Zhu}},\ }\bibfield
  {title} {\bibinfo {title} {A survey on tight euclidean t-designs and tight
  relative t-designs in certain association schemes},\ }\href@noop {}
  {\bibfield  {journal} {\bibinfo  {journal} {Proceedings of the Steklov
  Institute of Mathematics}\ }\textbf {\bibinfo {volume} {288}},\ \bibinfo
  {pages} {189} (\bibinfo {year} {2015})}\BibitemShut {NoStop}%
\bibitem [{\citenamefont {Bannai}\ and\ \citenamefont
  {Bannai}(2009{\natexlab{a}})}]{bannai2009survey}%
  \BibitemOpen
  \bibfield  {author} {\bibinfo {author} {\bibfnamefont {E.}~\bibnamefont
  {Bannai}}\ and\ \bibinfo {author} {\bibfnamefont {E.}~\bibnamefont
  {Bannai}},\ }\bibfield  {title} {\bibinfo {title} {A survey on spherical
  designs and algebraic combinatorics on spheres},\ }\href@noop {} {\bibfield
  {journal} {\bibinfo  {journal} {European Journal of Combinatorics}\ }\textbf
  {\bibinfo {volume} {30}},\ \bibinfo {pages} {1392} (\bibinfo {year}
  {2009}{\natexlab{a}})}\BibitemShut {NoStop}%
\bibitem [{Note1()}]{Note1}%
  \BibitemOpen
  \bibinfo {note} {Otherwise, we can make $K$ and $d_E$ arbitrarily large if we
  can use the entire phase-space.}\BibitemShut {Stop}%
\bibitem [{\citenamefont {Gottesman}(1996)}]{Gottesman_1996}%
  \BibitemOpen
  \bibfield  {author} {\bibinfo {author} {\bibfnamefont {D.}~\bibnamefont
  {Gottesman}},\ }\bibfield  {title} {\bibinfo {title} {Class of quantum
  error-correcting codes saturating the quantum hamming bound},\ }\href
  {https://doi.org/10.1103/physreva.54.1862} {\bibfield  {journal} {\bibinfo
  {journal} {Physical Review A}\ }\textbf {\bibinfo {volume} {54}},\ \bibinfo
  {pages} {1862–1868} (\bibinfo {year} {1996})}\BibitemShut {NoStop}%
\bibitem [{\citenamefont {Aly}(2007)}]{aly2007notequantumhammingbound}%
  \BibitemOpen
  \bibfield  {author} {\bibinfo {author} {\bibfnamefont {S.~A.}\ \bibnamefont
  {Aly}},\ }\href {https://arxiv.org/abs/0711.4603} {\bibinfo {title} {A note
  on quantum hamming bound}} (\bibinfo {year} {2007}),\ \Eprint
  {https://arxiv.org/abs/0711.4603} {arXiv:0711.4603 [quant-ph]} \BibitemShut
  {NoStop}%
\bibitem [{\citenamefont {Roy}\ and\ \citenamefont
  {Suda}(2014)}]{roy2014complex}%
  \BibitemOpen
  \bibfield  {author} {\bibinfo {author} {\bibfnamefont {A.}~\bibnamefont
  {Roy}}\ and\ \bibinfo {author} {\bibfnamefont {S.}~\bibnamefont {Suda}},\
  }\bibfield  {title} {\bibinfo {title} {Complex spherical designs and codes},\
  }\href@noop {} {\bibfield  {journal} {\bibinfo  {journal} {Journal of
  Combinatorial Designs}\ }\textbf {\bibinfo {volume} {22}},\ \bibinfo {pages}
  {105} (\bibinfo {year} {2014})}\BibitemShut {NoStop}%
\bibitem [{\citenamefont {Bannai}\ \emph {et~al.}(2021)\citenamefont {Bannai},
  \citenamefont {Bannai}, \citenamefont {Ito},\ and\ \citenamefont
  {Tanaka}}]{bannai2021algebraic}%
  \BibitemOpen
  \bibfield  {author} {\bibinfo {author} {\bibfnamefont {E.}~\bibnamefont
  {Bannai}}, \bibinfo {author} {\bibfnamefont {E.}~\bibnamefont {Bannai}},
  \bibinfo {author} {\bibfnamefont {T.}~\bibnamefont {Ito}},\ and\ \bibinfo
  {author} {\bibfnamefont {R.}~\bibnamefont {Tanaka}},\ }\href@noop {} {\emph
  {\bibinfo {title} {Algebraic combinatorics}}},\ Vol.~\bibinfo {volume} {5}\
  (\bibinfo  {publisher} {Walter de Gruyter GmbH \& Co KG},\ \bibinfo {year}
  {2021})\BibitemShut {NoStop}%
\bibitem [{\citenamefont {Bajnok}(2006)}]{bajnok2006euclidean}%
  \BibitemOpen
  \bibfield  {author} {\bibinfo {author} {\bibfnamefont {B.}~\bibnamefont
  {Bajnok}},\ }\bibfield  {title} {\bibinfo {title} {On euclidean designs},\
  }\href {https://doi.org/doi:10.1515/ADVGEOM.2006.026} {\bibfield  {journal}
  {\bibinfo  {journal} {Advances in Geometry}\ }\textbf {\bibinfo {volume}
  {6}},\ \bibinfo {pages} {423} (\bibinfo {year} {2006})}\BibitemShut {NoStop}%
\bibitem [{\citenamefont {Bannai}\ and\ \citenamefont
  {Bannai}(2009{\natexlab{b}})}]{BannaiBannai2009SphericalEuclidean}%
  \BibitemOpen
  \bibfield  {author} {\bibinfo {author} {\bibfnamefont {E.}~\bibnamefont
  {Bannai}}\ and\ \bibinfo {author} {\bibfnamefont {E.}~\bibnamefont
  {Bannai}},\ }\bibfield  {title} {\bibinfo {title} {Spherical designs and
  euclidean designs},\ }in\ \href@noop {} {\emph {\bibinfo {booktitle} {Recent
  Developments in Algebra and Related Areas}}},\ \bibinfo {series} {Advanced
  Lectures in Mathematics}, Vol.~\bibinfo {volume} {8},\ \bibinfo {editor}
  {edited by\ \bibinfo {editor} {\bibfnamefont {C.}~\bibnamefont {Dong}}\ and\
  \bibinfo {editor} {\bibfnamefont {F.-a.}\ \bibnamefont {Li}}}\ (\bibinfo
  {publisher} {Higher Education Press and International Press},\ \bibinfo
  {address} {Beijing and Somerville, MA},\ \bibinfo {year} {2009})\ pp.\
  \bibinfo {pages} {1--37}\BibitemShut {NoStop}%
\bibitem [{\citenamefont {Bajnok}(2007)}]{bajnok2007orbits}%
  \BibitemOpen
  \bibfield  {author} {\bibinfo {author} {\bibfnamefont {B.}~\bibnamefont
  {Bajnok}},\ }\bibfield  {title} {\bibinfo {title} {Orbits of the
  hyperoctahedral group as euclidean designs},\ }\href@noop {} {\bibfield
  {journal} {\bibinfo  {journal} {Journal of Algebraic Combinatorics}\ }\textbf
  {\bibinfo {volume} {25}},\ \bibinfo {pages} {375} (\bibinfo {year}
  {2007})}\BibitemShut {NoStop}%
\bibitem [{\citenamefont {Denys}\ and\ \citenamefont
  {Leverrier}(2023)}]{denys20232}%
  \BibitemOpen
  \bibfield  {author} {\bibinfo {author} {\bibfnamefont {A.}~\bibnamefont
  {Denys}}\ and\ \bibinfo {author} {\bibfnamefont {A.}~\bibnamefont
  {Leverrier}},\ }\bibfield  {title} {\bibinfo {title} {The $2 t $-qutrit, a
  two-mode bosonic qutrit},\ }\href@noop {} {\bibfield  {journal} {\bibinfo
  {journal} {Quantum}\ }\textbf {\bibinfo {volume} {7}},\ \bibinfo {pages}
  {1032} (\bibinfo {year} {2023})}\BibitemShut {NoStop}%
\bibitem [{\citenamefont {Lescanne}\ \emph {et~al.}(2020)\citenamefont
  {Lescanne}, \citenamefont {Villiers}, \citenamefont {Peronnin}, \citenamefont
  {Sarlette}, \citenamefont {Delbecq}, \citenamefont {Huard}, \citenamefont
  {Kontos}, \citenamefont {Mirrahimi},\ and\ \citenamefont
  {Leghtas}}]{lescanne_exponential_2020}%
  \BibitemOpen
  \bibfield  {author} {\bibinfo {author} {\bibfnamefont {R.}~\bibnamefont
  {Lescanne}}, \bibinfo {author} {\bibfnamefont {M.}~\bibnamefont {Villiers}},
  \bibinfo {author} {\bibfnamefont {T.}~\bibnamefont {Peronnin}}, \bibinfo
  {author} {\bibfnamefont {A.}~\bibnamefont {Sarlette}}, \bibinfo {author}
  {\bibfnamefont {M.}~\bibnamefont {Delbecq}}, \bibinfo {author} {\bibfnamefont
  {B.}~\bibnamefont {Huard}}, \bibinfo {author} {\bibfnamefont
  {T.}~\bibnamefont {Kontos}}, \bibinfo {author} {\bibfnamefont
  {M.}~\bibnamefont {Mirrahimi}},\ and\ \bibinfo {author} {\bibfnamefont
  {Z.}~\bibnamefont {Leghtas}},\ }\bibfield  {title} {\bibinfo {title}
  {Exponential suppression of bit-flips in a qubit encoded in an oscillator},\
  }\href {https://doi.org/10.1038/s41567-020-0824-x} {\bibfield  {journal}
  {\bibinfo  {journal} {Nature Physics}\ }\textbf {\bibinfo {volume} {16}},\
  \bibinfo {pages} {509} (\bibinfo {year} {2020})}\BibitemShut {NoStop}%
\bibitem [{\citenamefont {Albert}\ \emph {et~al.}(2018)\citenamefont {Albert},
  \citenamefont {Noh}, \citenamefont {Duivenvoorden}, \citenamefont {Young},
  \citenamefont {Brierley}, \citenamefont {Reinhold}, \citenamefont {Vuillot},
  \citenamefont {Li}, \citenamefont {Shen}, \citenamefont {Girvin} \emph
  {et~al.}}]{albert2018performance}%
  \BibitemOpen
  \bibfield  {author} {\bibinfo {author} {\bibfnamefont {V.~V.}\ \bibnamefont
  {Albert}}, \bibinfo {author} {\bibfnamefont {K.}~\bibnamefont {Noh}},
  \bibinfo {author} {\bibfnamefont {K.}~\bibnamefont {Duivenvoorden}}, \bibinfo
  {author} {\bibfnamefont {D.~J.}\ \bibnamefont {Young}}, \bibinfo {author}
  {\bibfnamefont {R.}~\bibnamefont {Brierley}}, \bibinfo {author}
  {\bibfnamefont {P.}~\bibnamefont {Reinhold}}, \bibinfo {author}
  {\bibfnamefont {C.}~\bibnamefont {Vuillot}}, \bibinfo {author} {\bibfnamefont
  {L.}~\bibnamefont {Li}}, \bibinfo {author} {\bibfnamefont {C.}~\bibnamefont
  {Shen}}, \bibinfo {author} {\bibfnamefont {S.~M.}\ \bibnamefont {Girvin}},
  \emph {et~al.},\ }\bibfield  {title} {\bibinfo {title} {Performance and
  structure of single-mode bosonic codes},\ }\href@noop {} {\bibfield
  {journal} {\bibinfo  {journal} {Physical Review A}\ }\textbf {\bibinfo
  {volume} {97}},\ \bibinfo {pages} {032346} (\bibinfo {year}
  {2018})}\BibitemShut {NoStop}%
\bibitem [{\citenamefont {Nielsen}\ and\ \citenamefont
  {Chuang}(2010)}]{nielsen2010quantum}%
  \BibitemOpen
  \bibfield  {author} {\bibinfo {author} {\bibfnamefont {M.~A.}\ \bibnamefont
  {Nielsen}}\ and\ \bibinfo {author} {\bibfnamefont {I.~L.}\ \bibnamefont
  {Chuang}},\ }\href@noop {} {\emph {\bibinfo {title} {Quantum computation and
  quantum information}}}\ (\bibinfo  {publisher} {Cambridge university press},\
  \bibinfo {year} {2010})\BibitemShut {NoStop}%
\bibitem [{\citenamefont {Li}\ and\ \citenamefont
  {Su}(2023)}]{li2023correctingbiasednoiseusing}%
  \BibitemOpen
  \bibfield  {author} {\bibinfo {author} {\bibfnamefont {Z.}~\bibnamefont
  {Li}}\ and\ \bibinfo {author} {\bibfnamefont {D.}~\bibnamefont {Su}},\ }\href
  {https://arxiv.org/abs/2308.01549} {\bibinfo {title} {Correcting biased noise
  using gottesman-kitaev-preskill repetition code with noisy ancilla}}
  (\bibinfo {year} {2023}),\ \Eprint {https://arxiv.org/abs/2308.01549}
  {arXiv:2308.01549 [quant-ph]} \BibitemShut {NoStop}%
\bibitem [{\citenamefont {H{\"a}nggli}\ \emph {et~al.}(2020)\citenamefont
  {H{\"a}nggli}, \citenamefont {Heinze},\ and\ \citenamefont
  {K{\"o}nig}}]{hanggli2020enhanced}%
  \BibitemOpen
  \bibfield  {author} {\bibinfo {author} {\bibfnamefont {L.}~\bibnamefont
  {H{\"a}nggli}}, \bibinfo {author} {\bibfnamefont {M.}~\bibnamefont
  {Heinze}},\ and\ \bibinfo {author} {\bibfnamefont {R.}~\bibnamefont
  {K{\"o}nig}},\ }\bibfield  {title} {\bibinfo {title} {Enhanced noise
  resilience of the surface--gottesman-kitaev-preskill code via designed
  bias},\ }\href@noop {} {\bibfield  {journal} {\bibinfo  {journal} {Physical
  Review A}\ }\textbf {\bibinfo {volume} {102}},\ \bibinfo {pages} {052408}
  (\bibinfo {year} {2020})}\BibitemShut {NoStop}%
\bibitem [{\citenamefont {Stafford}\ and\ \citenamefont
  {Menicucci}(2023)}]{Stafford_2023}%
  \BibitemOpen
  \bibfield  {author} {\bibinfo {author} {\bibfnamefont {M.~P.}\ \bibnamefont
  {Stafford}}\ and\ \bibinfo {author} {\bibfnamefont {N.~C.}\ \bibnamefont
  {Menicucci}},\ }\bibfield  {title} {\bibinfo {title} {Biased
  gottesman-kitaev-preskill repetition code},\ }\bibfield  {journal} {\bibinfo
  {journal} {Physical Review A}\ }\textbf {\bibinfo {volume} {108}},\ \href
  {https://doi.org/10.1103/physreva.108.052428} {10.1103/physreva.108.052428}
  (\bibinfo {year} {2023})\BibitemShut {NoStop}%
\bibitem [{\citenamefont {Zhang}\ \emph {et~al.}(2023)\citenamefont {Zhang},
  \citenamefont {Wu},\ and\ \citenamefont {Guo}}]{Zhang_2023}%
  \BibitemOpen
  \bibfield  {author} {\bibinfo {author} {\bibfnamefont {J.}~\bibnamefont
  {Zhang}}, \bibinfo {author} {\bibfnamefont {Y.-C.}\ \bibnamefont {Wu}},\ and\
  \bibinfo {author} {\bibfnamefont {G.-P.}\ \bibnamefont {Guo}},\ }\bibfield
  {title} {\bibinfo {title} {Concatenation of the gottesman-kitaev-preskill
  code with the xzzx surface code},\ }\bibfield  {journal} {\bibinfo  {journal}
  {Physical Review A}\ }\textbf {\bibinfo {volume} {107}},\ \href
  {https://doi.org/10.1103/physreva.107.062408} {10.1103/physreva.107.062408}
  (\bibinfo {year} {2023})\BibitemShut {NoStop}%
\bibitem [{\citenamefont {Berent}\ \emph {et~al.}(2024)\citenamefont {Berent},
  \citenamefont {Hillmann}, \citenamefont {Eisert}, \citenamefont {Wille},\
  and\ \citenamefont {Roffe}}]{Berent_2024}%
  \BibitemOpen
  \bibfield  {author} {\bibinfo {author} {\bibfnamefont {L.}~\bibnamefont
  {Berent}}, \bibinfo {author} {\bibfnamefont {T.}~\bibnamefont {Hillmann}},
  \bibinfo {author} {\bibfnamefont {J.}~\bibnamefont {Eisert}}, \bibinfo
  {author} {\bibfnamefont {R.}~\bibnamefont {Wille}},\ and\ \bibinfo {author}
  {\bibfnamefont {J.}~\bibnamefont {Roffe}},\ }\bibfield  {title} {\bibinfo
  {title} {Analog information decoding of bosonic quantum low-density
  parity-check codes},\ }\bibfield  {journal} {\bibinfo  {journal} {PRX
  Quantum}\ }\textbf {\bibinfo {volume} {5}},\ \href
  {https://doi.org/10.1103/prxquantum.5.020349} {10.1103/prxquantum.5.020349}
  (\bibinfo {year} {2024})\BibitemShut {NoStop}%
\bibitem [{\citenamefont {Chamberland}\ \emph {et~al.}(2022)\citenamefont
  {Chamberland}, \citenamefont {Noh}, \citenamefont {Arrangoiz-Arriola},
  \citenamefont {Campbell}, \citenamefont {Hann}, \citenamefont {Iverson},
  \citenamefont {Putterman}, \citenamefont {Bohdanowicz}, \citenamefont
  {Flammia}, \citenamefont {Keller}, \citenamefont {Refael}, \citenamefont
  {Preskill}, \citenamefont {Jiang}, \citenamefont {Safavi-Naeini},
  \citenamefont {Painter},\ and\ \citenamefont {Brandão}}]{Chamberland_2022}%
  \BibitemOpen
  \bibfield  {author} {\bibinfo {author} {\bibfnamefont {C.}~\bibnamefont
  {Chamberland}}, \bibinfo {author} {\bibfnamefont {K.}~\bibnamefont {Noh}},
  \bibinfo {author} {\bibfnamefont {P.}~\bibnamefont {Arrangoiz-Arriola}},
  \bibinfo {author} {\bibfnamefont {E.~T.}\ \bibnamefont {Campbell}}, \bibinfo
  {author} {\bibfnamefont {C.~T.}\ \bibnamefont {Hann}}, \bibinfo {author}
  {\bibfnamefont {J.}~\bibnamefont {Iverson}}, \bibinfo {author} {\bibfnamefont
  {H.}~\bibnamefont {Putterman}}, \bibinfo {author} {\bibfnamefont {T.~C.}\
  \bibnamefont {Bohdanowicz}}, \bibinfo {author} {\bibfnamefont {S.~T.}\
  \bibnamefont {Flammia}}, \bibinfo {author} {\bibfnamefont {A.}~\bibnamefont
  {Keller}}, \bibinfo {author} {\bibfnamefont {G.}~\bibnamefont {Refael}},
  \bibinfo {author} {\bibfnamefont {J.}~\bibnamefont {Preskill}}, \bibinfo
  {author} {\bibfnamefont {L.}~\bibnamefont {Jiang}}, \bibinfo {author}
  {\bibfnamefont {A.~H.}\ \bibnamefont {Safavi-Naeini}}, \bibinfo {author}
  {\bibfnamefont {O.}~\bibnamefont {Painter}},\ and\ \bibinfo {author}
  {\bibfnamefont {F.~G.}\ \bibnamefont {Brandão}},\ }\bibfield  {title}
  {\bibinfo {title} {Building a fault-tolerant quantum computer using
  concatenated cat codes},\ }\bibfield  {journal} {\bibinfo  {journal} {PRX
  Quantum}\ }\textbf {\bibinfo {volume} {3}},\ \href
  {https://doi.org/10.1103/prxquantum.3.010329} {10.1103/prxquantum.3.010329}
  (\bibinfo {year} {2022})\BibitemShut {NoStop}%
\bibitem [{\citenamefont {Fukui}\ \emph {et~al.}(2023)\citenamefont {Fukui},
  \citenamefont {Matsuura},\ and\ \citenamefont {Menicucci}}]{Fukui_2023}%
  \BibitemOpen
  \bibfield  {author} {\bibinfo {author} {\bibfnamefont {K.}~\bibnamefont
  {Fukui}}, \bibinfo {author} {\bibfnamefont {T.}~\bibnamefont {Matsuura}},\
  and\ \bibinfo {author} {\bibfnamefont {N.~C.}\ \bibnamefont {Menicucci}},\
  }\bibfield  {title} {\bibinfo {title} {Efficient concatenated bosonic code
  for additive gaussian noise},\ }\bibfield  {journal} {\bibinfo  {journal}
  {Physical Review Letters}\ }\textbf {\bibinfo {volume} {131}},\ \href
  {https://doi.org/10.1103/physrevlett.131.170603}
  {10.1103/physrevlett.131.170603} (\bibinfo {year} {2023})\BibitemShut
  {NoStop}%
\bibitem [{\citenamefont {Le~R{\'e}gent}\ \emph {et~al.}(2023)\citenamefont
  {Le~R{\'e}gent}, \citenamefont {Berdou}, \citenamefont {Leghtas},
  \citenamefont {Guillaud},\ and\ \citenamefont {Mirrahimi}}]{le2023high}%
  \BibitemOpen
  \bibfield  {author} {\bibinfo {author} {\bibfnamefont {F.-M.}\ \bibnamefont
  {Le~R{\'e}gent}}, \bibinfo {author} {\bibfnamefont {C.}~\bibnamefont
  {Berdou}}, \bibinfo {author} {\bibfnamefont {Z.}~\bibnamefont {Leghtas}},
  \bibinfo {author} {\bibfnamefont {J.}~\bibnamefont {Guillaud}},\ and\
  \bibinfo {author} {\bibfnamefont {M.}~\bibnamefont {Mirrahimi}},\ }\bibfield
  {title} {\bibinfo {title} {High-performance repetition cat code using fast
  noisy operations},\ }\href@noop {} {\bibfield  {journal} {\bibinfo  {journal}
  {Quantum}\ }\textbf {\bibinfo {volume} {7}},\ \bibinfo {pages} {1198}
  (\bibinfo {year} {2023})}\BibitemShut {NoStop}%
\bibitem [{\citenamefont {Xu}\ \emph {et~al.}(2023)\citenamefont {Xu},
  \citenamefont {Zheng}, \citenamefont {Wang}, \citenamefont {Zoller},
  \citenamefont {Clerk},\ and\ \citenamefont {Jiang}}]{xu2023autonomous}%
  \BibitemOpen
  \bibfield  {author} {\bibinfo {author} {\bibfnamefont {Q.}~\bibnamefont
  {Xu}}, \bibinfo {author} {\bibfnamefont {G.}~\bibnamefont {Zheng}}, \bibinfo
  {author} {\bibfnamefont {Y.-X.}\ \bibnamefont {Wang}}, \bibinfo {author}
  {\bibfnamefont {P.}~\bibnamefont {Zoller}}, \bibinfo {author} {\bibfnamefont
  {A.~A.}\ \bibnamefont {Clerk}},\ and\ \bibinfo {author} {\bibfnamefont
  {L.}~\bibnamefont {Jiang}},\ }\bibfield  {title} {\bibinfo {title}
  {Autonomous quantum error correction and fault-tolerant quantum computation
  with squeezed cat qubits},\ }\href@noop {} {\bibfield  {journal} {\bibinfo
  {journal} {npj Quantum Information}\ }\textbf {\bibinfo {volume} {9}},\
  \bibinfo {pages} {78} (\bibinfo {year} {2023})}\BibitemShut {NoStop}%
\bibitem [{\citenamefont {Korolev}\ \emph {et~al.}(2024)\citenamefont
  {Korolev}, \citenamefont {Bashmakova},\ and\ \citenamefont
  {Golubeva}}]{Korolev_2024}%
  \BibitemOpen
  \bibfield  {author} {\bibinfo {author} {\bibfnamefont {S.~B.}\ \bibnamefont
  {Korolev}}, \bibinfo {author} {\bibfnamefont {E.~N.}\ \bibnamefont
  {Bashmakova}},\ and\ \bibinfo {author} {\bibfnamefont {T.~Y.}\ \bibnamefont
  {Golubeva}},\ }\bibfield  {title} {\bibinfo {title} {Error correction using
  squeezed fock states},\ }\bibfield  {journal} {\bibinfo  {journal} {Quantum
  Information Processing}\ }\textbf {\bibinfo {volume} {23}},\ \href
  {https://doi.org/10.1007/s11128-024-04549-w} {10.1007/s11128-024-04549-w}
  (\bibinfo {year} {2024})\BibitemShut {NoStop}%
\bibitem [{\citenamefont {Gutman}\ \emph {et~al.}(2025)\citenamefont {Gutman},
  \citenamefont {Blumenthal}, \citenamefont {Hacohen-Gourgy}, \citenamefont
  {Orda},\ and\ \citenamefont
  {Kaminer}}]{gutman2025squeezedvacuumbosoniccodes}%
  \BibitemOpen
  \bibfield  {author} {\bibinfo {author} {\bibfnamefont {N.}~\bibnamefont
  {Gutman}}, \bibinfo {author} {\bibfnamefont {E.}~\bibnamefont {Blumenthal}},
  \bibinfo {author} {\bibfnamefont {S.}~\bibnamefont {Hacohen-Gourgy}},
  \bibinfo {author} {\bibfnamefont {A.}~\bibnamefont {Orda}},\ and\ \bibinfo
  {author} {\bibfnamefont {I.}~\bibnamefont {Kaminer}},\ }\href
  {https://arxiv.org/abs/2511.06108} {\bibinfo {title} {Squeezed-vacuum bosonic
  codes}} (\bibinfo {year} {2025}),\ \Eprint {https://arxiv.org/abs/2511.06108}
  {arXiv:2511.06108 [quant-ph]} \BibitemShut {NoStop}%
\bibitem [{\citenamefont {Cahill}\ and\ \citenamefont
  {Glauber}(1969)}]{cahill1969ordered}%
  \BibitemOpen
  \bibfield  {author} {\bibinfo {author} {\bibfnamefont {K.~E.}\ \bibnamefont
  {Cahill}}\ and\ \bibinfo {author} {\bibfnamefont {R.~J.}\ \bibnamefont
  {Glauber}},\ }\bibfield  {title} {\bibinfo {title} {Ordered expansions in
  boson amplitude operators},\ }\href@noop {} {\bibfield  {journal} {\bibinfo
  {journal} {Physical Review}\ }\textbf {\bibinfo {volume} {177}},\ \bibinfo
  {pages} {1857} (\bibinfo {year} {1969})}\BibitemShut {NoStop}%
\end{thebibliography}%
